\documentclass[11pt,reqno]{amsart}
\usepackage{amssymb}
\usepackage{amsthm}
\usepackage[dvipsnames]{xcolor}

\usepackage{enumerate}
\usepackage{graphicx}
\usepackage[dvips]{psfrag}
\usepackage[active]{srcltx}
\usepackage{amsaddr}

\newtheorem{theorem}{Theorem}[section]

\newtheorem{proposition}{Proposition}

\theoremstyle{definition}
\newtheorem{definition}[theorem]{Definition}
\newtheorem{remark}{Remark}

\renewcommand{\div}{{\mbox{div}}}

\advance\textwidth by +1.0in \advance\textheight by +0.5in
\advance\oddsidemargin by -0.5in \advance\evensidemargin by -0.5in
\advance\topmargin by -0.5in

\usepackage[english]{babel}

\title[Integral Aspects of 3D Lotka--Volterra Equations with $(1:-3:1)$ Resonance]
{Integral Aspects of $(1:-3:1)$ Resonance Lotka--Volterra Equations with Nonstandard Analysis}

\author[Chiman Qadir, Waleed Aziz, Ibrahim Hamad]
{Chiman Qadir$^1$, Waleed Aziz$^1$, Ibrahim Hamad$^1$}

\address{$^1$ Department of Mathematics, College of Science, Salahaddin
University--Erbil,
Kurdistan region--Iraq}

\email{chiman.qadir@su.edu.krd, waleed.aziz@plymouth.ac.uk,ibrahim.hamad@su.edu.krd}


\begin{document}

\begin{abstract}	
In this paper the problems of integrable and linearizable Lotka-Volterra equations with $(\delta:-3 \delta:\delta)$-resonance are studied. The necessary conditions for both problems are obtained in the case when $\delta=1$ and it's sufficiency are proved. It is also shown that non-standard analysis has an important role in proving the sufficient integrable conditions in some cases. The non-standard analysis approaches have been used for both perturbed and unperturbed cases for integral aspects of the given system.
\end{abstract}

\maketitle
	%%%%%%%%%%%%%%%%%%%%%%%%%%%%%%%%%%%%%%%%%%%%%%%%%
\section{Introduction}\label{sec1}
The first integral aspects of differential equations are one of interesting topics in the qualitative theory of ordinary differential equations. The knowledge of this kind of integral gives important information about the behavior of orbits near the origin. We mean, here, this kind of problem behaves as a
local integrability problem of differential systems. 

One useful method to study the integrability problem, which  particularly plays an important role
in the study of the dynamics of polynomial differential systems is Darboux theory of integrability.
This method is based on using invariant algebraic surfaces to construct first integrals.

For systems of dimension $n$, the existence of less than $n-1$ functionally independent first
integrals means that the system is partially integrable, and the existence of $n-1$
functionally independent first integrals means that the system is completely integrable.

In this paper the problems of complete integrability and linearizability at the origin for a three dimensional Lotka-Volterra system
\begin{eqnarray}\label{ch}
\begin{split}
\dot{x}&=P=x(\delta+ax+by+cz),\\ \dot{y}&=Q=y(-3 \delta+dx+ey+fz),\\
\dot{z}&=R=z(\delta+gx+hy+kz),
\end{split}
\end{eqnarray}
are investigated where variables. The parameters are considered as nonstandard 
hyperreals where $\delta$ is
limited and bounded away from zero.
The system (\ref{ch}) and its associated vector field $\chi$ 
	\begin{equation}\label{V}
	\chi=P\frac{\partial }{\partial x} +Q\frac{\partial }{\partial y}+R\frac{\partial }{\partial z},
	\end{equation}
are analytic around the origin.
%%%%%%%%%%%%%%%%%%%%%%%%%%%%%%%%%%%%%%%%%%%%%%%%%%%%%%%%%

Many models in the nature can be related to Lotka--Volterra equations introduced individually by Alfred Lotka (1925) and Vito Volterra (1926). These kind of equations are emphasized in studying of population dynamics and in applied sciences as well. More recently, Lotka--Volterra equations
are based in models of COVID-19 \cite{covid1, covid2}. 
Particularly, Lotka--Volterra system was studied by several researchers. In two dimensional
systems, it was considered by Christopher and Rousseau \cite{16}, Laurant and Llibre \cite{7},
Liu \cite{15} and many others to find conditions for the existence of a local first integrals in a
neighborhood of the origin. Many authors generalized a similar concept from two dimensional
systems into dimension three. An investigation of necessary and sufficient conditions for
integrability and linearizability for a three dimensional Lotka--Volterra systems was done by Laurant
\cite {6, 7} and for the resonance $(1:-1:1)$, $(2:-1:1)$ and $(1:-2:1)$ by Aziz and Christopher \cite{5}
and $(1:3:-1)$ in \cite{10}.
The authors in \cite{WAC, W,13}, give necessary and sufficient conditions for a family of three
dimensional quadratic nonlinearities.\\

Nonstandard analysis is considered in many different fields in mathematics and particularity
in differential equations. Moreover, nonstandard analysis tools
enable us to approach the non-exact solution more precisely than that with classical methods.
That is, with nonstandard there is no mathematical restrictions to find solutions near
singularity and regarding it as an alternative to the exact solution or as a shadow of the exact
solution. Some previous and recent works of this technique are done by Sari \cite{18, 1} who gave a
nonstandard model for perturbation of differential equations, Lorby and Sari \cite{19} about uniform
infinitesimal boundedness of singular perturbation of differential equation. Hamad and Pirdawood
\cite{21} investigated nonstandard solution in the monad of singularity of differential.
Hamad, in \cite{22, 23} gives nonstandard generalization of Sierpinski Theorem and introduced
generalized curvature and torsion in nonstandard analysis with some other concepts in differential geometry. 

%%%%%%%%%%%%%%%%%%%%%%%%%%%%%%%%%%%%%%%%%%%%%%%%%%%%%%%

The aim of this paper is to study the integrability problem of Lotka-Volterra system in two
different aspects. We are first interested in determining the necessary and sufficient conditions for both
integrability and
linearizability problems for Lotka--Volterra system with $(1:-3:1)$ resonant singular point at the origin.
The second is to extend the studied problem to hyperreal. We consider the system \eqref{ch} with
$(\delta:-3 \delta:\delta)$ resonant so that $\delta =1+\varepsilon$ where $\varepsilon$ is
infinitesimal or $\delta$ is itself infinitesimal, we find for both eigenvalues $(1:-3:1)$ and $(1+\varepsilon:-3(1+\varepsilon):1+\varepsilon)$
 in which the vector fields are infinitely close to each other and have two independent first integrals which are
infinitely close to standard two first integrals where the variables are near standard points in a
perturbation case. The perturbation of vector fields was proved as a deformation in nonstandard
\cite{1}. 
In \cite{8,9}, Gao showed that the homogeneous case $\delta=0$ will reduce to a two dimensional Lotka--Volterra equations. 
Other nonstandard transform use for unperturbation case of the system, which leads to the same integrability and linearizability conditions
as in standard form. On the other hand, making an infinitesimal perturbation with infinitesimal $\delta$
leads to new different results.
If we conditionally treated parameter $a$, then the obtained two independent first integrals of the
system \eqref{ch} for $\delta=1+\varepsilon$, where $\varepsilon$ is infinitesimal, belong to
monad(0), galaxy(0) or monad($\omega$), where $\omega$ is unlimited. However, if $\delta$ itself is infinitesimal, then the obtained
two independent first integrals of the system \eqref{ch} belong to monad(0) or galaxy(0).

%%%%%%%%%%%%%%%%%%%%%%%%%%%%%%%%%%%%%%%%%%%%%%
\section{Basic Definitions and Known Results}
 In this section, some basic notions regarding integrability, linearizability problems and
 mathematical analysis are introduced. The first part is devoted to the emphasis of basic concepts
 and known results of integrability problem, and in the second part some background of mathematical
 analysis is described.\\
 %%%%%%%%%%%%%%%%%%%%%%%%%%%%%%%%
 
We recall that a non-constant analytic function $\phi(x,y,z)$ is a first integral of the differential system (\ref{ch}) if it is constant on all its solutions that satisfies 

	$$\chi(\phi)=P\frac{\partial \phi }{\partial x} +Q\frac{\partial \phi }{\partial y}+R\frac{\partial \phi }{\partial z}=0.$$
It is well known that the equation $\phi(x,y,z)=c$ for some constant $c\in \mathbb{R}$ provides
an implicit set of trajectories of the given system.

The system \eqref{ch} is said to be integrable at the origin if there exists a change of coordinates 
	$$X=x(1+O(x,y,z)),  \quad Y=y(1+O(x,y,z)), \quad  Z=z(1+O(x,y,z)),$$
bringing \eqref{ch} to a system orbitally equivalent to
\begin{eqnarray} \label{cch}
	\dot{X}=\delta X \kappa , \quad
	\dot{Y}=-3 \delta Y \kappa, \quad
	\dot{Z}=\delta Z \kappa
\end{eqnarray}
where $\kappa=\kappa(X,Y,Z)=1+O(X,Y,Z)$.\\
Thus two functionally independent first integrals of system (\ref{cch}) are 
      $$\Psi_{1}=X^{3}Y \quad \mbox{and} \quad   \Psi_{2}=Y Z^{3},$$
 and pull back to first integrals of system (\ref{ch}), we obtain the required functionally independent first integrals
 	$$\phi_{1}=x^{3}y(1+O(x,y,z))\quad \mbox{and}  \quad \phi_{2}=yz^{3}(1+O(x,y,z)).$$

%%%%%%%%%%%%%%%%%%%%%%%%%%%%%%%%%%%%%%%%%%%%%%%%%%%%%
Sometimes we can find an analytic change of variables that reduces the nonlinear system into
a linear system. Then the reduced system is said to be linearizable.
If the change of coordinates can be selected so that $\kappa=1$ in equation (\ref{cch}) then we say the system is linearizable.
%%%%%%%%%%%%%%%%%%%%%%%%%%%%%%%%%%%%%%%%%%%%%

Let $ l \in  \mathbb{C}[x,y,z]$ be a polynomial, then a surface $l=0$ is called an invariant algebraic surface of the system (\ref{ch}), if $l$ satisfies the linear partial differential equation
	$$\chi(l)=P\frac{\partial l }{\partial x} +Q\frac{\partial l }{\partial y}+R\frac{\partial l }{\partial z}=l C_{l},$$
for some polynomial $C_{l} \in  \mathbb{C}[x,y,z]$ are called the cofactor of the invariant algebraic surface $l=0$. It is well known that the degree of the cofactor is at most one.
%%%%%%%%%%%%%%%%%%%%%%%%%%%%%%%%%%%%%%%%%%%%%

A function of the form $E(x,y,z)=\exp(\frac{f}{g})$ where $f,g \in \mathbb{C}[x,y,z]$
is an exponential factor if it satisfies
$$\chi(E)=C_{E} E,$$
for some polynomial $C_{E}$ of degree at most one. The polynomial $C_{E}$ is called the cofactor of $E$.
 
%%%%%%%%%%%%%%%%%%%%%%%%%%%%%%%%%%%%%%%%%%%%%
If $l_i$ are invariant algebraic surfaces and $E=\exp(\frac{f}{g})$ is an exponential factor such that
there exist $\lambda_0, \lambda_i\in \mathbb{C}$ for $i=1,\ldots,n$, then the function of the form
$$\Phi=\prod l_{i}^{\lambda^{i}} E^{\lambda_{0} f/g },$$
is a Darboux function.
Given a Darboux function $\Phi$, we can compute 
	$$\chi(\Phi)= \Phi ( \sum \lambda_{i} C_{l_{i}}+\lambda_{0} C_{E}),$$ 
where $C_{l_i}$ and $C_E$ are respective cofactors of $l_i$ and $E$.
Obviously, the function $\Phi$ is a non-trivial first integral of the system if and only if the cofactors $C_{l_{i}}$ and $C_{E}$ are linearly dependent.

%%%%%%%%%%%%%%%%%%%%%%%%%%%%%%%%%%%%%%%%%%%%%
Note that the function $\mu$ is an inverse Jacobi multiplier for the vector field \eqref{V}
if it satisfies the equation
	$$\chi(\mu)= \mu \: \div(\chi).$$
In three dimensional system, if we construct one first integral $\Phi$ and an inverse Jacobi multiplier
$\mu$, then we can find a second independent first integral.

\begin{theorem}\cite{11}\label{w}
Suppose the analytic vector fields\\
		$$(\lambda+\sum\nolimits_{|I|>0} A_{xI}X^I) \frac{\partial }{\partial x}+y (\mu+\sum\nolimits_{|I|>0} A_{yI}X^I) \frac{\partial }{\partial y}+(v+\sum\nolimits_{|I|>0} A_{zI}X^I) \frac{\partial }{\partial z} ,$$	
has a first integral $\phi=x^{\alpha} y^{\beta}z^{\gamma}(1+O(x,y,z))$ with at least one of $\alpha,\beta,\gamma \neq 0$ and an inverse Jacobi multiplier $M=x^{r}y^{s}z^{t}(1+O(x,y,z))$ and suppose that the cross product of $(r-i-1,s-j-1,t-k-1)$ and $(\alpha,\beta,\gamma)$ is bounded away from zero for any integers $i,j,k \ge 0$,  then the system has a second analytic first integral of the form $\psi=x^{1-r} y^{1-s}z^{1-t}(1+O(x,y,z))$, and hence the system (\ref{ch}) is integrable.
\end{theorem} 
%%%%%%%%%%%%%%%%%%%%%%%%%%%%%%%%%%%%%%%%%%%

\begin{theorem}\cite{5}\label{w1}
If the system (\ref{ch}) with $\delta=1$ is integrable and there exists a function $ \zeta=x^{\alpha}y^{\beta}z^{\gamma} (1+O(x,y,z))$ such that $\chi(\zeta)=k \zeta$ for some constant $k=\alpha \lambda+\beta \mu + \gamma \nu$, then the system is linearizable.
\end{theorem}
%%%%%%%%%%%%%%%%%%%%%%%%%%%%%%%%%%%%%%

\begin{theorem}\cite{2}\label{colin}
	For $\lambda \in \mathbb{N} -\{1\} $ the Lotka--Volterra system
	\begin{eqnarray}\label{1}
	\dot{x}= x(1+c_{20}x + c_{11}y), \quad  \dot{y}= y(-\lambda+d_{11}x+d_{02}y),
	\end{eqnarray}
	has a linearizable saddle at the origin if and only if one of the following
	conditions is satisfied\\
	\begin{enumerate}[i.]
	\item $ m c_{20}+d_{11}=0 , m = 0,\ldots,\lambda -2$.
	\item $ c_{11}=d_{02} = 0$.
	\item $ c_{02} - d_{11} = c_{11} - d_{02} = 0$.
	\item $c_{11} = (\lambda - 1)c_{20} + d_{11} = 0.$
	\end{enumerate}
\end{theorem}
%%%%%%%%%%%%%%%%%%%%%%%%%%%%%%%%%%%%

\section{Poincar\'e-Dulac Normal Form Theory}
The theory of normal form of a nonlinear differential systems  first introduced in the work of
Poincar\'e. This theory is crucial in studying the integral aspects of differential systems and gives a clue of general prospective of the form of their integrals. So, here, we give main facts
of normal forms and some notations as well. 

Let us now consider an analytic system 
\begin{equation} \label{zhang1}
\dot x=Ax+F(x),
\end{equation}
where $A$ is an $n\times n$ matrix with eigenvalues $\alpha_{1}, \ldots, \alpha_{n}$,
$x \in \mathbb{R}^n$ and $F(x)=O(2)$ is an analytic vector-valued function without constant or
linear term in $(\mathbb{R}^n,0)$.

\begin{definition}\cite{poincare}
If the convex hull of the spectrum of matrix $\sigma(A)$ does not contain the origin, then  $\sigma(A)$
is said to be in the Poincar\'e domain. If the origin lies inside the convex hull of  $\sigma(A)$,
then we say  $\sigma(A)$ is in the Siegel domain.
\end{definition}

\begin{remark}
A singular point whose spectrum lies in the Poincar\'e domain can be brought to normal form by an analytic change of coordinates. Particularly, a node with two analytic separatrices can have no resonant terms in its normal form and so must be analytically linearizable \cite{5}.
\end{remark}

\begin{theorem}\cite{poincare}
If the spectrum of matrix is in the Poincar\'e domain, then there are at most finitely many
resonant monomials.
\end{theorem}

\begin{definition}
An n-tuple $\alpha=(\alpha_{1}, \ldots, \alpha_{n})$ of eigenvalues of $A$ is called {\it resonant}, if it satisfies the resonant identity
\begin{equation} \label{res}
\alpha_{i}=\langle \ell,\alpha \rangle=\sum^{n}_{i=1} \ell_{i} \alpha_{i}, \quad |\ell| \geq 2,
\end{equation}
for non-negative integers $\ell_{i} \in \mathbb{N}_{0} = \{0, 1, 2, \ldots \},$ and
the natural numbers $|\ell|= \sum^{n}_{i=1} \ell_{i}$
is the order of the resonance.
The coefficient $X^{\ell}_{i}$ of the monomial $x$ is called a resonant coefficient and the corresponding
term is called a resonant term.
\end{definition}

Note that if condition (\ref{res}) does not hold, then $\alpha_{1}, \ldots, \alpha_{n}$ are called non-resonant and then the differential equation can be formally linearized
(i.e. a formal, but in general, non-convergent transformation series exists).
For more details consult \cite{Ilyashenko2008,Romanovski2009}.\\

Let $\mathcal{H}_s$ denote the vector space of functions from $\mathbb{R}^n$ to $\mathbb{R}^n$
each of whose components is a homogeneous polynomial of degree $s$.

The linear operator (homological operator) $\mathcal{L}$ on $\mathcal{H}_s$ is defined by
\[\mathcal{L}h(y)=d\,h(y)\,Ay-Ah(y).\]

By performing a sequence of transformations of the form
\[x=H(y)=y+h^s(y),\qquad h^s\in\mathcal{H}_s,\]
we can remove all non-resonant terms of 
system \eqref{zhang1} and is formally equivalent to 
\begin{equation} \label{zhang2}
\dot y=Ay+G(y).
\end{equation}

The problem of when an analytic system is analytically equivalent to
its normal form is classical and is an open problem.

\begin{theorem}\cite{Romanovski2009}
Let $(\alpha_{1}, \ldots, \alpha_{n})$ be the eigenvalues of $A$ in \eqref{zhang1}. If the equality
\eqref{res} does not holds for all $i\in{1,\ldots,n}$ and for all $\alpha \in \mathbb{N}_{0}$ in which 
$|\alpha|\geq0$, then \eqref{zhang1} and \eqref{res} are formally equivalent for all $F$ and $G$.

\end{theorem}

Let
\[\mathcal{M}_\alpha=\Big\{ \ell=(\ell_1,\ldots,\ell_n): \langle \ell,\alpha \rangle=\sum^{n}_{i=1} \ell_i \alpha_i =0, \quad |\ell| \geq 1 \Big\},\]
 be a set of vectors and let $R_\alpha$ be the rank of vectors in the set $\mathcal{M}_\alpha$.
It is obvious  $R_\alpha\leq n-1$.

\begin{theorem}\cite{X.Zhang}\label{ThZhang}
Suppose that the origin of system \eqref{zhang1} is non-degenerate. Then system \eqref{zhang1} has
$n-1$ linearly independent analytic first integrals if and only if $R_\lambda=n-1$,
and it is analytically equivalent to its normal form
\[\dot y_i=\lambda_i y_i(1+G(y)), \quad i=1, \ldots,n,\]
by an analytic normalization, where $G(x)$ has no constant term and is an analytic
function of $y^\ell$ with $\ell\in \mathcal{M}$ and $\ell=(\ell_1,\ldots,\ell_n)=1$.
\end{theorem}

For more detail on the formal and analytic transformations to the
normal form using Poincar\'e-Dulac normal form, one can consult
\cite{F. Dumortier2006,W.Walcher}.

%%%%%%%%%%%%%%%%%%%%%%%%%%%%%%%%%%%%%%%%%%%%%%%
\section{Mathematical Analysis and Nonstandard backgrounds}
We briefly recall some mathematical analysis concepts and results related to the problem in this paper.\\
A subset $S$ of a metric space $(\mathsf{X},d)$ is compact if every
sequence from $S$ has a sub sequence which converges to an element of $S$. If $\mathsf{X}$
is compact, we say the metric space itself is compact. 
%%%%%%%%%%%%%%%%%%%%%%%%%%%%%%%%%%%%%%%%%%%%%
\begin{theorem}[Heine-Borel Theorem]\cite{12}\label{v}
Every closed and bounded subset of $\mathbb{R}^{n}$ is compact.
\end{theorem}

%%%%%%%%%%%%%%%%%%%%%%%%%%%%%%%%%%%%%%%%%%%%%

Hyperreal numbers are an extension of the real numbers, which contain infinitesimals and unlimited numbers. The set of hyperreal numbers is denoted by ${}^{*} \mathbf{R}$ and for simplicity we use  $\mathbf{R}$ without the star.\\
A number $x \in \mathbf{R}$  is unlimited or infinite if $|x|> n$ for all standard $n \in \mathbb{N} $, $x $ is  limited or finite if $|x|< n$ for some standard $n \in \mathbb{N} $ and  $x  $ is infinitesimal if $|x|< \frac{1}{n}$ for all standard $n \in \mathbb{N} $. Note that the number zero is the only standard  infinitesimal. A real number is appreciable if it is neither unlimited nor infinitesimal. A real number $x$ is called near-standard if $x$ is infinitely close to $y$ for some standard real number $y$, and denoted by $^{NS}$.
	
The two real numbers $x$ and $y$ are infinitely close  if $x-y$ is infinitesimal and is written as $x \simeq y$.\\

A number $x$ is limited in $\mathbf{R}$, then it is infinitely close to a unique standard real number; this unique number is called the standard part of $x$. It is also well known that
\begin{enumerate}
\item $\mbox{monad}(x)=\{y \in \mathbf{R}: x \simeq y\} $ for limited $x$, and is dented by $\mbox{m}(x)$.

\item $\alpha$-monad $(x)=\{y \in \mathbf{R}:\frac{y-x}{\alpha}  \simeq 0 \} $, and denoted by  $\alpha$-{m}(x).
	
 \item galaxy$(x)=\{y \in \mathbf{R}: x - y  \ \  is  \ \ \mbox{limited} \} $ for limited $x$, and is dented by $\mbox{gal}(x)$.

\item $\alpha$-galaxy$(x)=\{y \in \mathbf{R}:\frac{y-x}{\alpha} \ \ is \ \  \mbox{limited} \} $, and dented by $\alpha$-{gal}(x). 
 
\item $\alpha$-micromonad $(x)=\{y \in \mathbf{R}:y \leq {\alpha} ^{n}, \mbox{for all standard $n$}\}$ .(see \cite{4, 20}).

\end{enumerate}
%%%%%%%%%%%%%%%%%%%%%%%%%%%%%%%%%%%%%

  %%%%%%%%%%%%%%%%%%%%%%%%%%%%%%%%%%%%
 \begin{definition}\cite{1}
 Let  $F:D \to P$ and  $F_{0}: D_{0} \to P$ be mappings from the open subsets  $ D$ and $ D_{0}$ of
 the standard topological space $\mathsf{X}$ to the uniform space $P$, $F_{0}$ standard. The mapping $F$ is
 said to be perturbation of the mapping $F_{0}$, which is denoted by $F\simeq F_{0} $, if $ {}^{N S}
D_{0} \in D$ and   $F(x)\simeq F_{0}(x) $ for all $x \in {}^{N S}D_{0} $, where $x \in {}^{N S}D_{0} $
is $\{x\in \mathsf{X}: \exists^{st} x_0\in D_0, x \simeq x_0 \}$.
\end{definition}

This definition make a sense, because $D_{0} $ is being a standard open subset of $\mathsf{X}$, ${}^{NS}D_{0} \subset D$, so $F(x)$ and $F_{0}(x)$ are both defined for all $ x\in {}^{N S}D_{0} $. Let $ \Omega_{\mathsf{X},P}$ be the set of mappings defined on an open subsets of $\mathsf{X}$ to $P$, that is 
$$\Omega_{\mathsf{X},P}=\{(F,D):D\: \mbox{open subset of}\: \mathsf{X}\: \mbox{and} \: F:D \to P\}.$$
The topology of this set is defined as follows. Let $(F_{0},D_{0}) \in \Omega_{\mathsf{X},P}$. The family of sets of the form 
\[
\{(F,D) \in \Omega_{\mathsf{X},P}:K \subset D  \quad \forall x \in K \quad   (F(x),F_{0}(x)) \in U \}
\]
where $K$ is a compact subset of $D_{0}$ and $U$ is an entourage of the uniform space $P$ is a basis of the set of neighborhood of $(F_{0}, D_{0})$. This topology is known as topology of uniform convergence on compacta. If all the mappings are defined on the same open set $D$, this topology is the usual topology of uniform convergence on compacta on the set of functions on $D$ to $P$. For more details, see \cite{4,Joshi,3,1,Nader}.

%%%%%%%%%%%%%%%%%%%%%%
\begin{proposition}\cite{1}\label{sari}
 Assume $\mathsf{X}$ is locally compact. The mapping $F$ is a perturbation of the standard mapping $F_{0}$ if and only if $F$ is infinitely close to $F_{0}$ for the topology of uniform convergence on compacta.
\end{proposition}

%%%%%%%%%%%%%%%%%%%%%%%%%%%%%%%%%%%%%%%5555
\section{Computation of necessary conditions of integrability and linearizability}
The considering problem in this section is the problem of integrals.
This problem and other closely related problems, here, can be reduced to investigate the
polynomial ideals and its associate varieties.

In this section, we study the integrability and linearizability conditions at the origin of system (\ref{ch})
with $\delta=1$ and the variables and parameters are real numbers. The necessary conditions were
found by computing the conditions of the existence of two independent first integrals up to a degree.

Here, we are looking for two functionally independent first integrals 
\begin{equation}\label{phi}
\phi=x^3y+\ldots,
\end{equation}
and
\begin{equation}\label{psi}
\psi=z^3y+\ldots,
\end{equation}
which satisfies respectively, 
\begin{equation}
\chi\phi=\sum \mathcal{R}_\phi X,
\end{equation}

\begin{equation}
\chi\psi=\sum \mathcal{R}_\psi Y,
\end{equation}
where $X$ and $Y$ are monomial resonance and  $\mathcal{R}_\phi$ and $\mathcal{R}_\psi$ are
polynomials in the coefficients of the given system, which are obstacle of the existence of first
integrals \eqref{phi} and \eqref{psi}. The vanishing of these obstacles gives necessary and sufficient
conditions for integrability conditions.
Hence, the first integrals $\phi$ and $\psi$ exist whenever all $\mathcal{R}_\phi$ and
$\mathcal{R}_\psi$ are zero.
Moreover, let $\mathcal{B}$ be the ideal generated by $\mathcal{R}_\phi$ and $\mathcal{R}_\psi$.
We compute the variety of the ideal $\mathcal{B}$, and let us denote it by $V(\mathcal{B})$,
and then the irreducible decomposition of $V(\mathcal{B})$, variety of integral needs to be computed.
Note that the set of $\mathcal{R}_\phi$ and $\mathcal{R}_\psi$ involving infinitely many polynomials
and finding $\mathcal{B}$, in this case, is impossible. In fact, by the aid of Hilbert Basic Theorem
every ideal in the polynomial ring is finitely generated. Hence, the ideal $\mathcal{B}$ will stabilize
at some finite order. The degrees 16 of resonant monomials were used using computer system Maple. The irreducible decomposition were found by using the {\tt minAssGTZ} in computer algebra
system {\sc Singular} \cite{lib,DGPS}.

For linearizing, we use the same method of  computing the conditions for the existence of a
linearizing change of coordinates up to some finite order to find necessary conditions,
and exhibiting a linearizing change of variables for sufficiency.\\

Hence, this gives the following theorem.

Consider the three dimensional Lotka-Volterra system 
\begin{eqnarray}\label{2}
F_{0}=\left\{\begin{array}{rc} 
\dot{x}&=F_{01}=x(1+ax+by+cz),\\ \dot{y}&=F_{02}=y(-3+dx+ey+fz),\\
 \dot{z}&=F_{03}=z(1+gx+hy+kz), \end{array}\right.
\end{eqnarray}

where $F_{0}=(F_{01},F_{02},F_{03})$, and all variables and parameters are standard numbers.

\begin{theorem}\label{int}
The three dimensional Lotka-Volterra system \eqref{2} is integrable at origin if and only if the system satisfies each of the following conditions
\begin{equation*}
\begin{aligned}
		 	 1)  \: & g=f=d=c=0\\
		 	  2)  \: & f=d=c=g+a=0\\ 
	         	2^*)  \: & g=f=d=c+k=0 \\
		     3)  \: & g=f=d=a=0 \\
		      3^*) \: & f=d=c=k=0\\
		     4)  \: & h=g=f=b-e=a-d=0\\
		      4^*)  \: & d=c=b=f-k=e-h=0\\
		     5)  \: & f=e+h=d+g=c=b+h=a+g=0\\
		      5^*)  \: & g=f-k=e-h=d=c+k=b+h=0\\
		     6)  \: & f=2e+h=2d+g=c=2b+h=2a+g=0 \\
		     6^*)  \: & g=d=f-k=e-h=c+2k=b+2h=0\\
		     7)  \: & f=d=b-h=0\\
		     8)  \: & f=5e-3h=d+g=c=5b+h=3a-g=0\\
		     9) \: & f=d=c-k=a-g=0\\
		     10) \: & f=8e-3h=4d+3g=c=8b+h=4a-g=0\\
		     11) \: & f=8e-3h=4d+3g=c+2k=8b+h=4a-g=0\\
		     12) \: & k=f=c=b=2a+d=0 \\
		     12^*)  \: & h=g=f+2k=d=a=0\\
		     13) \: & h=g=f=b=c+k=a+d=0 \\
		      13^*) \: & h=f+k=d=c=b=a+g=0\\
		     14) \: & h=f=c=b=2a+d=0\\
		      14^*) \: & h=g=d=b=f+2k=0\\
		      15) \: & f=d+g=c=b=2a-g=0 \\
		      15^*) \: & h=g=d=f+2k=c-2k=0\\
		     16) \: & f=c=b=3d+2g=3a-g=0 \\
		      16^*) \: & h=g=d=c-3k=f+2k=0\\
		     17) \: & f=3d+2g=c+k=b=3a-g=0 \\
		      17^*) \: & h=f+2k=d=c-3k=a+g=0\\
\end{aligned}
\end{equation*}
\begin{equation*}
\begin{aligned}
		     18) \: & g=f=c=a+d=0 \\
		      18^*) \: & g=f+k=d=c=0\\
		     19) \: & f=k=c=a+d=0 \\
		     19^*) \: & g=d=a=f+k=0\\
		        20)  \: & f=c=b-h=a+d=0 \\
		     20^*) \: & g=f+k=d=b-h=0\\
		    21) \: & f=2d+g=c=2a-g=0 \\
		     21^*) \: & g=f+k=d=c-2k=0\\
		    22) \: & f=e+h=3d+g=2c-k=b=3a-g=0 \\
		     22^*) \: & h=f+k=c-3k=b+e=a-2g=d=0\\
		    23) \: & f+k=e+h=d+g=c-3k=b+h=a+g=0 \\
		     23^*) \: & f-k=e-h=3d+g=c+k=b+h=3a-g=0\\
		    24) \: & f+k=5e-3h=d+g=c+k=5b+h=3a-g=0\\
		    25) \: & h=g=f+k=c=b=a+d=0\\
		    26) \: & h=g=f+2k=c-2k=a+d=0\\
		    26^*) \: & f+k=d+g=c=b=2a-g=0 \\
		    27) \: & h=f+k=3d+2g=c=b=3a-g=0 \\
		     27*) \: & h=g=b=f+2k=c-3k=a+d=0\\
		    28) \: & h=f+k=3d+2g=2c-k=b=3a-g=0 \\
		     28^*) \: & h=f+2k=d+2g=c-3k=b=a-2g=0\\
		     29) \: & h=g=f+k=b-e=a-d=0 \\
		      29^*) \: &  f-k=e-h=c=b=a+d=0\\
		      30) \: & h=g=f+2k=b-e=a-d=0 \\
		     30^*) \: & f-k=e-h=c=b=2a+d=0\\
		      31) \: & f+k=e=3d+g=c-3k=b+h=3a-g=0\\
		    32) \: & cd-2dk-gk=f+k=b-h=a+d=0\\
		     33) \: & f+k=d+g=c-k=a-g=0\\
		    34) \: & h=e=b\\
		    35) \: & cd-2dk-2gk=h=f+2k=b=2a+d=0\\
		    36) \: & h=f+2k=3d+2g=c-3k=b=3a-g=0\\
		    37) \: & g=f+3k=e+3h=d=c-3k=b+5h=0\\
		    38) \: & g=f+3k=e+3h=d=c-4k=b+8h=0\\
\end{aligned}
\end{equation*}
\begin{equation*}
\begin{aligned}	 
		    39) \: & f+3k=e+3h=d=c-4k=b+8h=2a+g=0\\
		   40) \: & f+3k=e+3h=d-g=c-3k=b+5h=a+g=0\\
		    41) \: & ef+2ek-3hk=cd+3cg-2dk+fg-3gk=bk-ce+ek-hk\\&=bf+2bk-3ch-fh+hk=bd+3bg-de-dh-2eg=af+3ak-dk-3gk\\&=ae+3ah-de-3eg=ac-2ak-cd-3cg+2dk-fg+4gk=ab+2ah-bd-3bg+dh=0.
\end{aligned}
\end{equation*}	 
Furthermore, the conditions $(1)-(7),(9),(12)-(23),(25)-(36)$ are satisfied if and only if the system is
linearizable or one of the following holds
\begin{equation*}
\begin{aligned}
		 41.1\: &  dh-eg=k=f=c=b-e=a-d=0\\
		  41.1^* \: & bk-ch=g=f-k=e-h=d=a=0\\
		   41.2 \: & f-k=e-h=d-g=c-k=b-h=a-g=0.\\
\end{aligned}
\end{equation*}	 		 
 \end{theorem}
\begin{proof}\	The dual cases are obtained under the transformation \((x,z)\)$\longrightarrow$ \((z,x)\) and we do not consider them further. The other cases are considered below.\\
	
%%%%%%%%%%%%%%%%%%%%%%%%%%%%%%%%%%%%%%%%%%%
\textbf{Case 1}: In this case, the system takes the form 
\begin{equation} \label{3}
\dot{x}=x(1+ax+by),\qquad \dot{y}=y(-3+ey),\qquad  \dot{z}=z(1+hy+kz).
\end{equation}
The change of coordinate $Y=\frac{y}{1-\frac{e}{3}y}$ linearizes the second equation to $\dot{Y}=-3Y$. To linearize the third equation, we seek an invariant algebraic surface of the form $\alpha(Y)+\beta(Y)z=0$ with cofactor $hy+kz$, where $\alpha$ and $\beta$ are analytic so that $\alpha(0)=1$. Thus, $\alpha$ and $\beta$ must satisfy the equations
\begin{equation} \label{4}
-3Y	\dot{\alpha}=h y  \alpha,
\end{equation} 
\begin{equation*} \label{5}
	3Y	\dot{\beta}-\beta=-k \alpha.
\end{equation*} 
The function $\alpha=(1+\frac{e}{3}Y)^{(-\frac{h}{e})} =\sum\nolimits_{n \ge 0} \eta Y^n$
\hspace{0.1cm} with $ |\frac{e}{3}Y |<1$ where $ \eta= \frac{(-1)^{n} (\frac{e}{3})^n \prod_{i=-1}^{n-1}
(\frac{h}{e}+i) }{n! (\frac{h}{e}-1)}$ satisfied equation (\ref{4}). It can also be seen that
$\beta=\sum\nolimits_{n \ge 0} \frac{- \eta  k  }{3n-1} Y^{n-\frac{1}{3}}+c_{1}Y^{\frac{1}{3}} $.
Similarly to linearize the first equation of (\ref{3}), we proceed as before by seeking an invariant
algebraic surface of the form $\tilde{\alpha}(Y)+\tilde{\beta}(Y)x=0$ with cofactor $ax+by$.
Finally, the change of variables $(X,Z)=(\frac{x}{\tilde{\alpha}(Y)+\tilde{\beta}(Y)x},\frac{z}{\alpha(Y)+\beta(Y)z})$, will linearize the first and third equations.\\

%%%%%%%%%%%%%%%%%%%%%%%%%%%%%%%%%%%%%%%%%%%
\textbf{Case 2}: The system reduces to
\begin{equation}\label{7}
				\dot{x}=x(1+ax+by-kz),\qquad \dot{y}=y(-3+ey),\qquad \dot{z}=z(1+hy+kz).
\end{equation}  
To linearize the system, the second and third equations are the same as equations in previous case,
so it suffices to linearize the first equation only. We investigate the existence of an invariant algebraic
surface $ A(Y,Z)+B(Y,Z)x=0$ with cofactor $ax+by-kz$, where $A$ and $B$ are analytic with
$A(0,0)=1$. The change of coordinate $ X=\frac{x}{A(Y,Z)+B(Y,Z)x}$ will linearize the first equation.
To find such $A$ and $ B$, they must satisfy the equations
		\begin{equation} \label{8}
				\dot{A}=(by-kz)  A,
		\end{equation} 
			\begin{equation} \label{9}
		\dot{B}+B=a A.
		\end{equation} 
	We write $A=\exp(\alpha(Y,Z))$ and put in equation(\ref{8}), we get 
	$\dot{\alpha}(Y,Z)=by(Y,Z)-kz(Y,Z)$ and 
	let  $\alpha=\sum\nolimits_{i+j \ge 0} \alpha_{ij} Y^i Z^{j} $.\\
Then $\dot{\alpha}(Y,Z)=\sum\nolimits_{i+j \ge 0} \alpha_{ij}(j-3i) Y^i Z^{j}=by(Y,Z)-kz(Y,Z)=\sum\nolimits_{i+j \ge 0} \gamma_{ij} Y^i Z^{j}$, therefore $\alpha_{ij}=\frac{\gamma_{ij}}{j-3i}
$ for $i+j>0$ and the convergence of $\sum\nolimits_{i+j \ge 0} \gamma_{ij} Y^i Z^{j}$ imply the
convergence of $\alpha$ and hence of $A$. If $A=\sum\nolimits_{i+j \ge 0} a_{ij} Y^i Z^{j}$ and
$B=\sum\nolimits_{i+j \ge 0} b_{ij} Y^i Z^{j}$ then from  equation (\ref{9}), we get
$B=\sum\nolimits_{i+j \ge 0} 
	 a \ a_{ij}/(j-3i-1) Y^i Z^{j}$ which is convergent if $j\neq 3i+1$,
then the system is linearizable.\\
	 
		%%%%%%%%%%%%%%%%%%%%%%%%%%%%%%%%%%%%%%
\textbf{Case 3}: \  Under the conditions of integrability, the system reduces to
		\begin{equation*}
			\dot{x}=x(1+by+cz),\qquad \dot{y}=y(-3+ey),\qquad \dot{z}=z(1+hy+kz). 
		\end{equation*}  
In this case, we try to linearize first equation and the rest are similar as in Case 1. The change of variable $X=x \exp(-A)$ will give $\dot{X}=X$ whenever $A$ satisfies
	\begin{equation}\label{11}
 \dot{A}(Y,Z) =  by(Y)+c z(Y,Z).
	\end{equation}  
We put  $A=\sum\nolimits_{i+j \ge 0} a_{ij} Y^i Z^{j} $. Then $\dot{A}(Y,Z)=\sum\nolimits_{i+j \ge 0} a_{ij}(3i-j) Y^i Z^{j}$, and it is obvious  equation (\ref{11}) can be solved
analytically for $A$ with $z(Y,Z)$ contains no term of the type $(YZ^{3})^{n}$.
Since $z(Y,Z)=\frac{Z \alpha(Y)}{1-Z \beta(Y)}=\sum \alpha \beta^{i-1} Z^{i}$.\\
If we let $i-1=3m$, then it is enough to show that $\alpha \beta^{3m-1}$ contains no term in $Y^{m}$.
Simple calculation shows
   	$$\alpha \beta^{3m-1}=\sum\nolimits_{n \ge 0} \eta Y^n \sum\nolimits_{p \ge 0} C_{p}^{3m-1}(c_{1})^{-(p+3m-1)} (\sum\nolimits_{n \ge 0} -\frac{\eta k}{3n-1} Y^n)^{p}  Y^{m-\frac{p-1}{3}},$$
which does not contain any term of  $Y^{m}$ and we have established the existence of linearizing change of coordinates.\\
   		
	%%%%%%%%%%%%%%%%%%%%%%%%%%%%%%%%%%%%%%
\textbf{Case 4}:   The system has the form 
\begin{equation*}
\dot{x}=x(1+dx+ey+cz),\qquad \dot{y}=y(-3+dx+ey),\qquad  \dot{z}=z(1+kz).
\end{equation*}
The change of coordinates $X=\frac{x}{1+dx-\frac{e}{3}y}$ and  $Y=\frac{y}{1+dx-\frac{e}{3}y}$ gives 
\begin{equation*} 
	\dot{X}=X(1-cz+dczX),\qquad \dot{Y}=Y(-3-dezX),\qquad \dot{z}=z(1+kz).
\end{equation*} 
The first and third equations are linearizable node and the change of coordinates are
$\tilde{X}=\tilde{X}(X,z)$ and $Z=Z(z)$. For linearizing the second equation, it is enough  to find 
$\Gamma(\tilde{X},Z)$  such that $\dot{\Gamma}=d\:e\:z\:X$ and the change of coordinate
$\tilde{Y}=Y e^{\Gamma}$ gives $\dot{\tilde{Y}}=-3 \tilde{Y}$. Let   $\Gamma=\sum\nolimits_{i+j \ge 0} \gamma_{ij} \tilde{X}^i Z^j$ \  and \ $X(\tilde{X},Z)z(Z)=\sum\nolimits_{i+j \ge 0} \eta_i \tilde{X}^{i} Z^{j} $.\\
We find that
	$$\dot{\Gamma}(\tilde{X},Z)=\sum\nolimits_{i+j \ge 0} (i+j) \gamma_{ij} \tilde{X}^i Z^j= \sum\nolimits_{i+j \ge 0} d e \eta_i \tilde{X}^{i} Z^{j},$$
Therefore, we see that $\gamma_{ij}=\frac{de\eta_{ij}}{i+j}$, which guaranties a convergent for $\Gamma$.\\	

%%%%%%%%%%%%%%%%%%%%%%%%%%%%%%%%%%%%%%%%%%%%%%%%%
\textbf{Case 5}: The conditions in this case yield the system
\[
	\dot{x}=x(1-gx-hy),\qquad  \dot{y}=y(-3-gx-hy),\qquad \dot{z}=z(1+gx+hy+kz).
\]
The system admits invariant algebraic surfaces $\l_1 =1+kz+\frac{1}{2} kz(hy-gx)=0$ and  $\l_2 =1-gx+\frac{1}{3}hy=0 $ with respective cofactors $kz$ and $-gx-hy$. Then the transformation
$X=x \hspace{0.1cm} l^{-1}_{2},Y=y \hspace{0.1cm} l^{-1}_{2},Z=z \hspace{0.1cm} l^{-1}_{1} \hspace{0.1cm} l_{2}$ linearizes the system to $	\dot{X}=X ,\dot{Y}=-3Y ,\dot{Z}=Z $.\\

%%%%%%%%%%%%%%%%%%%%%%%%%%%%%%%%%%%%%%%
\textbf{Case 6}:  The system becomes
	\begin{equation*}
		\dot{x}=x(1+dx+ey),\qquad  \dot{y}=y(-3+dx+ey),\qquad \dot{z}=z(1-2dx-2ey+kz).
	\end{equation*}
In this case we have two invariant algebraic surfaces $\l_1 =1+dx-\frac{1}{3} ey=0$ with cofactor
$dx+ey$ and  $\l_2 =1+kz(1+ dx-ey+\frac{1}{2}(dx)^{2}+\frac{1}{5}(ey)^{2})=0$ with cofactor $kz$,
which allow us to linearize the system by the substitution
	$(X,Y,Z)=(x \hspace{0.1cm} l^{-1}_{1},y \hspace{0.1cm} l^{-1}_{1},z \hspace{0.1cm} l^{2}_{1} \hspace{0.1cm} l_{2}^{-1})$.\\
	
%%%%%%%%%%%%%%%%%%%%%%%%%%%%%%%
\textbf{Case 7}: The integrability conditions imply the system 
\begin{equation*}
\dot{x}=x(1+ax+hy+cz),\qquad \dot{y}=y(-3+ey),\qquad \dot{z}=z(1+gx+hy+kz).
\end{equation*}
The change of variables $X=x \ n(y)$ and $Z=z \ n(y)$ for some analytic function $n(y)$ with
$n(0)=1$, then the system takes the form
\begin{equation*}
\dot{X}=X(1+aXn^{-1}+cZn^{-1}+hy+\frac{\dot{n}}{n} y(-3+ey)), \\
\end{equation*}
\begin{equation}\label{16}
\begin{aligned}
\dot{y}= y(-3+ey), \\ 
\end{aligned}
\end{equation}
\begin{equation*}
\dot{Z}=Z(1+gXn^{-1}+kZn^{-1}+hy+\frac{\dot{n}}{n} y(-3+ey)).
\end{equation*}
We can select $n$ such that $1+hy+\frac{\dot{n}}{n} y(-3+ey)=\frac{1}{n}$ with $n(0)=1$.\\
This equation, clearly, is an ordinary first order linear differential equation
	$$ \dot{n}-\frac{1}{3} \frac{1+hy}{y(1-\frac{e}{3}y)}n=-\frac{1}{3y(1-\frac{e}{3}y)}.$$
Therefor, its solution is
	$$n=-\frac{1}{3} e^{\int -\frac{1+hy}{3y(1-\frac{e}{3}y)}dy} ( \int  e^{\int \frac{1+hy}{3y(1-\frac{e}{3}y)}dy}   y^{-1}(1-\frac{e}{3} y)^{-1}  dy+C).$$
Manipulating some calculations gives
\begin{equation*}
\begin{aligned}
	n&=-\frac{1}{3}y^{-\frac{1}{3}}(1-\frac{e}{3})^{\frac{h}{e}+\frac{1}{3}} \int y^{-\frac{2}{3}}(1-\frac{e}{3}y)^{\frac{h}{e}+\frac{4}{3}}dy\\
	&=-\frac{1}{3}(1-\frac{e}{3}y)^{\frac{h}{e}+\frac{1}{3}} \sum_{i=0}^{\infty}(-1)^{i} \frac{\prod_{j=0}^{i} (\frac{4}{3}+\frac{h}{e}-(j-1))}{i!(\frac{4}{3}+\frac{h}{e}+1)(i+\frac{1}{3}))} (-\frac{e}{3}y)^{i},
\end{aligned}
\end{equation*}
which can be seen to be analytic in $y$ with $n(0)=1$. Then  rescaling it by $\frac{1}{n(y)}$, the system (\ref{16}) becomes 
\begin{equation*}
	{X'}=X(1+aX+cZ), \qquad
	{y'}=y(-3+ey) n(y), \qquad
	{Z'}=Z(1+gX+kZ).
\end{equation*}
where $ '=\frac{d}{d \tau} $ and $t=n \tau$ .The first and third equations gives a linearizable node. For linearizing the second equation, we substitute $Y=m(y)$, such that
	$$y(-3+ey)n(y) \frac{dm(y)}{dy}=-3m(y),\quad   m(0)=0,\quad m'(0)=1,$$
put $n(y)$ in equation above then it is analytically solvable .\\
%%%%%%%%%%%%%%%%%%%%%%%%%%%

\textbf{Case 8}: The system takes the form
\begin{equation*}
\dot{x}=x(1+ax+by),\qquad \dot{y}=y(-3-3ax-3by),\qquad \dot{z}=z(1+3ax-5by+kz),
\end{equation*}
which has two independent first integrals $\phi_1=x^{3} y$ and $\phi_2=y z^{3} l^{-3}$, where $\l=(1+ax+by)^{2}+kz(1+\frac{1}{2} ax)-by=0$ is an invariant algebraic surface with cofactor $2ax-6by+kz$.\\
%%%%%%%%%%%%%%%%%%%%%%%%%%%%%%

\textbf{Case 9}: Here the corresponding system is 
\begin{eqnarray}\label{19}
\dot{x}=x(1+ax+by+kz),\quad \hspace{0.1cm} \dot{y}=y(-3+ey),\quad \hspace{0.1cm} \dot{z}=z(1+ax+hy+kz).
\end{eqnarray}
The system has an invariant algebraic curve  $\l = 1-\frac{1}{3}ey=0$ with cofactor $ey$.
Then there is a first integral $\Phi=xz^{-1}l^{\frac{h-b}{e}}$ and an inverse Jacobi multiplier $x^{2}y^{\frac{4}{3}}zl^{(
	\frac{2}{3}-{\frac{b}{e}})}$. Thus, the first integral will be of the form $\Psi=x^{-1}y^{\frac{-1}{3}}(1+O(x,y,z)) $. The desired first integrals are $\phi_1=\Psi^{-3}=x^{3} y (1+O(x,y,z)$  and $\phi_2=(\Phi\Psi)^{-3}=y z^{3} (1+O(x,y,z)$. However, if $ \xi=y(1-\frac{e}{3}y)^{-1}$  satisfies $\dot{\xi}=-3 \xi$, then the system (\ref{19}) must also be linearizable by Theorem \ref{w1}.\\

%%%%%%%%%%%%%%%%%%%%%%%%%%%%%%%%%%%
\textbf{Case 10}: The system of this case is written as 
	\begin{equation*}
		\dot{x}=x(1+ax+by),\qquad \dot{y}=y(-3-3ax-3by),\qquad \dot{z}=z(1+4ax-8by+kz).
	\end{equation*}
We construct two independent first integrals $\phi_1=x^{3} y$  and $\phi_2=y z^{3} l^{-3}$, where
$\l = (1+ax+by)^{3}-3a b^{2}x y^{2}+kz(1+ax-by-\frac{1}{3}(ax)^{2}-2a b x y-\frac{1}{5}(by)^{2})=0$
is invariant algebraic  surface with cofactor  $3ax-9by+kz$.\\

%%%%%%%%%%%%%%%%%%%%%%%%%%%%%%%%
\textbf{Case 11}: The system has the form 
\[
\dot{x}=x(1+ax+by-2kz),\qquad \dot{y}=y(-3-3ax-3by),\qquad  \dot{z}=z(1+4ax-8by+kz).
\]
This case has an invariant algebraic surface $\l = (1+ax+by)^{3}+kz(3ax+2(ax)^{2}-6a b x y+akxz)=0$ with cofactor  $3ax-9by$ which yields the Darboux integral $\Phi=x y z^{2} l^{-2}$ and an inverse Jacobi multiplier $x y^{\frac{4}{3}} z^{2}l^{\frac{2}{3}}$. Therefor the other first integral is  of the form
$\Psi=y^{\frac{-1}{3}} z^{-1}(1+O(x,y,z))$ by Theorem \ref{w}.
The two independent first integrals will be desired as $\phi_1=\Phi^{3} \Psi^{6}=x^{3} y (1+\cdots)$  and $\phi_2=(\Psi)^{-3}=y z^{3}(1+\cdots)$.\\

%%%%%%%%%%%%%%%%%%%%%%%%%%%%%%%%%%
\textbf{Case 12}: The system in this case is 
	\begin{equation*}\
		\dot{x}=x(1+ax),\qquad \dot{y}=y(-3-2ax+ey),\qquad \dot{z}=z(1+gx+hy).
	\end{equation*}
The linearizing change is given by  $X= x l^{-1}_{1} $, $Y= y l^{3}_{1} l^{-1}_{2}$ and $Z=z l^{({\frac{h}{e}}-\frac{g}{a})}_{1} l^{-\frac{h}{e}}_{2}$, where $\l_{1} = 1+ax=0$ with cofactor $ax$ and $\l_{2} =1+ax-\frac{1}{3}ey=0$ with cofactor $ax+ey$ are invariant algebraic surfaces.\\

%%%%%%%%%%%%%%%%%%%%%%%%%%%%%%%%
\textbf{Case 13}: The reduced system is 
\begin{equation*}
\dot{x}=x(1+ax-kz),\qquad \dot{y}=y(-3-ax+ey), \qquad \dot{z}=z(1+kz).
\end{equation*}
The change of coordinates $(X,Y,Z)=(x \hspace{0.1cm} l^{-1}_{1} l_{3},y \hspace{0.1cm} l_{1} l^{-1}_{2} l_3,z \hspace{0.1cm} l^{-1}_{3} )$ will linearize the system where
 $\l_1=1+ax+\frac{1}{2}akxz$, $\l_2 =1-\frac{1}{3}ey+kz+\frac{1}{6}aexy-\frac{1}{3}ekyz$ and $l_{3}=1+kz$.\\
 
%%%%%%%%%%%%%%%%%%%%%%%%%%%%%%%%%%
\textbf{Case 14}: The system can be written in the form  
	\begin{equation*}
		\dot{x}=x(1+ax),\qquad  \dot{y}=y(-3-2ax+ey),\qquad  \dot{z}=z(1+gx+kz).
	\end{equation*}
The system has invariant surfaces $\l_{1} = 1+ax=0$, $\l_{2} =1+ax-\frac{1}{3}ey=0$ with cofactors
$ax$ and $ax+ey$, respectively.
The first and third equations are decoupled in the variables $x$ and $z$ which give linearizable
node and by the change of variables $X=x(1+O(x,z))$ and $Z=z(1+O(x,z))$ will linearize $\dot{X}=X$, $\dot{Z}=Z$.
The second equation is linearized by the substitution $Y= y l^{3}_{1}  l^{-1}_{2}$
which gives $\dot{Y}=-3Y$.\\

%%%%%%%%%%%%%%%%%%%%%%%%%%%%%%%%
\textbf{Case 15}: The system is of the form  
	\begin{equation*}
		\dot{x}=x(1+ax),\qquad \dot{y}=y(-3-2ax+ey),\qquad \dot{z}=z(1+2ax+hy+kz).
	\end{equation*}
This system will be transformed to 
 \begin{equation}\label{26}
	\dot{X}=X,\qquad  \dot{Y}=Y(-3+eY),\qquad  \dot{Z}=Z(1+hY+kZ),
\end{equation}
where $X=\frac{x}{1+ax}$, $Y=\frac{y}{1+ax}$ and $Z=\frac{z}{1+ax}$.\\The second equation of (\ref{26}) will be linearized by using the substitution $\hat{Y}=\frac{Y}{1-\frac{1}{3} e Y}$.
For linearizing third equation in (\ref{26}), the idea is the same in Case 1.\\

%%%%%%%%%%%%%%%%%%%%%%%%%%%%%%%%%%
\textbf{Case 16}: The system is of the form 
	\begin{eqnarray*}
		\dot{x}=x(1+ax),\qquad  \dot{y}=y(-3-2a x+ ey),\qquad \dot{z}.=z(1+3ax+hy+kz).
	\end{eqnarray*}
We have two invariant algebraic surfaces $l_{1}=1+ax=0$ and $l_{2}=1+ax-\frac{1}{3}ey=0$ with
respective cofactors $ax$ and $ax+ey$, respectively. The substitution $(X,Y)=(\frac{x}{1+ax},
\frac{y(1+ax)^{3}}{1+ax-\frac{1}{3}ey})$ will linearize the first and second equations to $\dot{X}=X $
and $\dot{Y}=-3Y$. For linearizing the third equation, we seek an algebraic surface of the form
$l=\gamma+\gamma_{1} z=0$ with cofactor $3ax+hy+kz$ with $\gamma_{1}=\gamma_{1}(X,Y)$
and $\gamma=\gamma(X,Y)$ such that $\gamma(0,0)=1 $. The form
$Z=\frac{z}{(\gamma+\gamma_{1} z)}$ will linearize the third one.
To find each of $\gamma$ and $\gamma_{1}$ we have to solve the following two differential equations 
	\begin{equation}\label{27}
		\dot{\gamma_{1}}-\gamma_{1} = k \gamma,
	\end{equation}
		\begin{equation}\label{28}
		\dot{\gamma} =(3ax+hy) \gamma.
	\end{equation}
We write $\gamma=e^{\eta(X,Y)}$ and from equation (\ref{28}) we obtain
\[
		 \dot{\eta}(X,Y)=3ax(X,Y)+hy(X,Y).
\]
Some calculations show
\[
		   y= \sum\nolimits_{n \ge 1} (-\frac{e}{3})^{n-1} Y^{n} (1-aX)^{3n+2},
\]
which contain no terms of the form $(X^3Y)^{n}$ because always $3n+2>3n$. \\
%%%%%%%%%%%%%%%%%%%%%%%%%%%%

\textbf{Case 17}: The system takes the form 
\begin{equation*}
	\dot{x}=x(1+ax-kz),\qquad \dot{y}=y(-3-2ax+ey),\qquad  \dot{z}=z(1+3ax+hy+kz).
\end{equation*}
In this case we can neither linearize the system nor obtain two independent first integrals  for
the system in classical way. So, if we use nonstandard approach like when $a \simeq 0$ implies
that $d \simeq 0$ and $g \simeq 0$, then the system becomes
\begin{equation*}
	\dot{x} \simeq x(1-kz),\qquad  \dot{y} \simeq y(-3+ey),\qquad  \dot{z} \simeq z(1+hy+kz).
\end{equation*}
The change of coordinate of $X=xz$ gives 
\begin{equation*}
	\dot{X} \simeq X(2+hy),\qquad  \dot{y} \simeq y(-3+ey),\qquad  \dot{z} \simeq z(1+hy+kz).
\end{equation*}
The system involvs an algebraic curve $l=1-\frac{e}{3}y=0$ with cofactor $ey$,
then the substitution  $(\hat{X},Y)=(X l^\frac{-h}{e},y l^{-1})$ will linearize first and second equations.
For linearizing third equation, it is the same as in Case 1 as considered before.\\

%%%%%%%%%%%%%%%%%%%%%%%%%%%%%%
\textbf{Case 18}: The system reduces to  
	\begin{equation*}
		\dot{x}=x(1+ax+by),\qquad \dot{y}=y(-3-ax+ey),\qquad \dot{z}=z(1+hy+kz).
	\end{equation*}
We can see that the first and second equations are linearizable saddle by Theorem \ref{colin}
part {\it i}. For linearizing the third equation, we seek an invariant algebraic surface of the form
$A+Bz=0$ with cofactor $hy+kz$ with $A=A(X,Y)$ and $B=B(X,Y)$ such that $A(0,0)=1 $.
Then the linearizable change is given by  $Z=\frac{z}{(A+B z)}$. To find such $A$ and $B$,
we have to solve 
\begin{equation}\label{30}
	\dot{A} =A hy,
\end{equation}
\begin{equation*}
\dot{B}+B = k A.
\end{equation*}
We now write $A=\exp(\eta(X,Y))$ and from equation (\ref{30}) we obtain $ \dot{\eta}=hy$.\\
To solve this equation, let $\eta= \sum\nolimits_{i+j>0} a_{ij} X^{i} Y^{j}$.\\ 
Then
 $ \dot{\eta}(X,Y)= \sum\nolimits_{i+j>0} (i-3j)a_{ij}X^{i} Y^{j}=hy(X,Y)=b_{0,1} hY+ \sum\nolimits_{i+j>1}b_{ij} X^{i} Y^{j}$, for some $b_{i,j}$. Note that $a_{1,0}=0, a_{0,1}=\frac{-h}{3} $ and $a_{ij}=\frac{b_{ij}}{(i-3j)}$ for $i+j>1$. The convergence of $\sum\nolimits_{i+j>0} b_{ij} X^{i} Y^{j}$ guarantees the convergence of $\eta$ and hence of $A$. It remains to show that the inverse transformation $z=\frac{AZ}{1-BZ}=\sum\nolimits AB^{n}Z^{n+1}$ does not involve the term $X^{3}Y$.\\
Suppose $n+1=3m$ for some $m$. It suffices to show that $AB^{3m-1}$ has no term $Y^{m}$.
Direct calculation shows
\begin{equation*}
\dot{B} =\frac{dB}{dt}=-3Y \frac{dB}{dY},
\end{equation*}
\begin{equation*}
\frac{-3Y}{3m} \frac{d B^{3m}}{dY}-B^{3m}=kAB^{3m-1}.
\end{equation*}
Clearly, the coefficient in $Y^{n} $ in $B^{3 m}$ vanishes on the right hand side so that either $k=0$, that mean $B\equiv 0$, or the coefficients of $Y^{n}$ in $B^{3m-1}$ vanishes. Therefore $Z=\frac{z}{A+Bz}$ contains no term like $(X^{3}Y)^n$.
Thus, we found the existence of a linearizing transformation.\\ 
%%%%%%%%%%%%%%%%%%%%%%%%%%%%

\textbf{Case 19}: The system in this case is 
	\begin{equation*} 
		\dot{x}=x(1+ax+by), \qquad \dot{y}=y(-3-ax+ey), \qquad \dot{z}=z(1+gx+hy).
	\end{equation*}
The first two equations give a linearizable saddle by Theorem \ref{colin}, part {\it i}
and hence $\dot{X}=X,\quad \dot{Y}=-3 Y$. \\
Using the substitution $Z= z \exp(-\gamma(X,Y))$ to linearize the remaining equation.
Now, it suffices to find $\gamma(X,Y)$ such that
\begin{equation} \label{33}
 \dot{\gamma}  =g x (X,Y)+h y(X,Y).
\end{equation} 
Let
 \begin{equation} \label{34}
\gamma(X,Y)= \sum\nolimits_{i+j>0} \alpha_{ij} X^{i} Y^{j},
 \end{equation} 
 then  
 	\begin{equation*} 
g x(X,Y)+hy(X,Y)= \sum\nolimits_{i+j>0} \beta_{ij} X^{i} Y^{j}.
 \end{equation*} 
 After differentiating (\ref{34}) and equating with (\ref{33}) we obtain
 	\begin{equation*} 
 \begin{aligned}
 \sum\nolimits_{i+j>0} (i-3j)\alpha_{ij}X^{i} Y^{j}=\sum\nolimits_{i+j>0} \beta_{ij} X^{i} Y^{j}\\ 
 \end{aligned} 
 \end{equation*}
   Then we can choose $\alpha_{ij}=\frac{\beta_{ij}}{(i-3j)}$ for $i\neq 3j$ and the convergence follows from the convergence of $\sum\nolimits_{i+j>0} \beta_{ij} X^{i} Y^{j}$. Thus the system is linearizable.\\
   
%%%%%%%%%%%%%%%%%%%%%%%%%%%%%%%%%%%%
\textbf{Case 20}:
The system with integrability conditions has the form 
\begin{eqnarray*}
	\dot{x}=x(1+ax+by),\qquad \dot{y}=y(-3-ax+ey),\qquad \dot{z}=z(1+gx+hy+kz).
\end{eqnarray*}
The first and second equations are linearizable saddle by Theorem \ref{colin}, part {\it i}.
The third equation can be  linearized in the similar way of Case 16.\\
%%%%%%%%%%%%%%%%%%%%%%%%%%%%%%%%%%%

\textbf{Case 21}: The system is written as
\begin{eqnarray*}
	\dot{x}=x(1+ax+by),\qquad \dot{y}=y(-3-ax+ey),\qquad  \dot{z}=z(1+2ax+hy+kz).
\end{eqnarray*}
The first two equations are linearizable saddle by Theorem \ref{colin}, part {\it i}.
The linearization of the third equation is  similar to the Case 16.\\
%%%%%%%%%%%%%%%%%%%%%%%%%%%%%%%%%%

\textbf{Case 22}: The corresponding system in this case is
\begin{equation*}
	\dot{x}=x(1+ax+cz),\qquad \dot{y}=y(-3-ax-hy),\qquad \dot{z}=z(1+3ax+hy+2cz),
\end{equation*}
 which has an invariant algebraic surface $\l = (1+ax)^{2}+cz(2+h y)=0$ with cofactor  $2(ax+cz) $.
We construct a first integral of the form $\Phi=x^{2} y z l^{-2}$ with an inverse Jacobi multiplier $IJM=x^{4} y^{2} z l^{\frac{-1}{2}}$, then the other first integral is  $\Psi=x^{-3} y^{-1} l^{\frac{-3}{2}} $
by Theorem \ref{w}.
Now, the required first integrals are $\phi_1=\Phi^{3} \Psi^{2}= y z^{3} (1+\cdots)$  and $\phi_2=(\Psi)^{-1}= x^{-1} y (1+\cdots)$. Since  $ \xi=xl^{-\frac{1}{2}}$  satisfies  $\dot{\xi}=  \xi$, then the system in this case must also be linearizable by Theorem \ref{w1}.\\
%%%%%%%%%%%%%%%%%%%%%%%%%%%%%%%%

\textbf{Case 23}: The system 
\begin{eqnarray*}
		\dot{x}=x(1-gx-hy+3kz), \hspace{0.3cm} \dot{y}=y(-3-gx-hy-kz), \hspace{0.3cm} \dot{z}=z(1+gx+hy+kz),
	\end{eqnarray*}
	 has an invariant algebraic surface $\l = (1+kz)^{2}+kz(hy-gx)=0$ with cofactor $ 2kz $.
	We find a first integral $\Phi=x y z^{2} l^{-2}$ and an inverse Jacobi multiplier $IJM=x y^{1/2} z^{-1/2} l$ and the other first integral 
	$\Psi= y^{\frac{1}{2}}z^\frac{3}{2}(1+...)$ is guaranteed by Theorem \ref{w}.
	  The two first integrals in standard form are $\phi_1=\Phi^{3} \Psi^{-4}=x^{3} y (1+\cdots)$  and $\phi_2=(\Psi)^{2}=  y z^{3} (1+\cdots)$. The substitution $(X,Y,Z)=(\frac{\Phi}{yz^{3}},\frac{\phi_{2}}{z^{3}},\frac{\Phi}{x y})$ linearizes the system.\\
%%%%%%%%%%%%%%%%%%%%%%%%%%%%%%

\textbf{Case 24}: Under these conditions the corresponding system has the form 
	\begin{eqnarray*}
		\dot{x}=x(1+ax+by-kz),\qquad \dot{y}=y(-3-3ax-3by-kz),\qquad \dot{z}=z(1+3ax-5by+kz).
	\end{eqnarray*}
	The system has two invariant algebraic surfaces $\l_{1} = (1+ax+by)^{2}+kz(ax-by)=0$ and $\l_{2}=(1+ax+by+kz)^{2}-4bkyz=0$ with respective cofactors $ 2ax-6by $ and $ 2ax-6by+2kz $.
	Then the system has two independent first integrals  $\phi_1=x^{3} y (l_{1}l_{2})^{2})$ and $\phi_2= y z^{3} l_{1}^{2} l_{2}^{-1}$.\\
%%%%%%%%%%%%%%%%%%%%%%%%%%%%%%%%

\textbf{Case 25}: The corresponding system is  
	\begin{eqnarray*}
		\dot{x}=x(1-dx),\qquad  \dot{y}=y(-3+dx+ey-kz),\qquad \dot{z}=z(1+kz).
	\end{eqnarray*}
The surfaces  $\l_{1}=1+kz=0$, $\l_{2}=1-dx=0$ and $\l_{3} =(1+kz)^{2}-\frac{1}{3}e y(1+\frac{1}{2}d x-\frac{3}{2}k z)=0$ are invariant algebraic surfaces with cofactors $kz$, $-dx$ and $ey+2kz$, respectively.\\
 By the change of coordinates $X= x l^{-1}_{2}$, $ Y= yl_{1}^{3}l_{2} l^{-1}_{3}$ and $Z=zl^{-1}_{1}$, we obtain $\dot{X}=X,\dot{Y}=-3Y,\dot{Z}=Z $.\\
%%%%%%%%%%%%%%%%%%%%%%%%%%%%%%%%%%

\textbf{Case 26}: In this case the system is
	\begin{eqnarray*}
		\dot{x}=x(1+ax+by+2kz), \qquad \dot{y}=y(-3-ax+ey-2kz), \qquad \dot{z}=z(1+kz).
	\end{eqnarray*}
By using  $X=\frac{x}{1+kz} $, $Y=\frac{y}{1+kz} $ and $Z=\frac{z}{1+kz}$, the system transforms to 
\begin{eqnarray*}
	\dot{X}=X(1+aX+bY),\hspace{0.5cm} \dot{Y}=Y(-3-aX+eY), \hspace{0.5cm} \dot{Z}=Z.
\end{eqnarray*}
where the first two equations are linearizable saddle by Theorem \ref{colin}, part {\it i} and then the system is linearizable.\\

%%%%%%%%%%%%%%%%%%%%%%%%%%%%%%%%%%
\textbf{Case 27}: The system takes the form
\begin{eqnarray*}
	\dot{x}=x(1+ax),\qquad \dot{y}=y(-3-2ax+ey-kz), \qquad \dot{z}=z(1+3ax+kz).
\end{eqnarray*}
The change of variable $Y=y l_{1} l_{2} l_{3} $ linearzies the second equation where $\l_{1} = (1+ax)^{2}+kz(1+\frac{1}{2} ax)=0$, $\l_{2} = 1+ax-\frac{1}{3}ey(1+kz)=0$ and $\l_{3}=1+ax=0 $ are invariant algebraic surfaces with respective cofactors $2ax+kz$, $ax+ey$ and  $ax$.
The other two equations are linearizable node and the change of coordinates $X=x(1+O(x,z))$ and $ Z=z(1+O(x,z)$) gives $\dot{X}=X,\dot{Z}=Z$.\\

%%%%%%%%%%%%%%%%%%%%%%%%%%%%%%%%%%%%%
\textbf{Case 28}: Under the conditions of this case, the system (\ref{2}) is written as  
	\begin{equation*}
		\dot{x}=x(1+ax+cz), \qquad \dot{y}=y(-3+2ax+ey-2cz), \quad \dot{z}=z(1+3ax+2cz).
	\end{equation*}
The system has two invariant algebraic surfaces  $\l_{1} = (1+ax)^{2}+2cz)=0$ and $\l_{2}=1+\frac{1}{3}ey(-1+ax+cz-(ax)^{2})=0$ with cofactors $2(ax+cz) $ and  $ey $, respectively.\\The second equation of the system is linearized by $Y=y l_{1} l_{2}^{-1} $ to $\dot{Y}=-3Y$ and the two other equations of the system  are linearizable node with the change of variables $X=x(1+O(x,z))$ and $ Z=z(1+O(x,z)$) so that $\dot{X}=X,\:\dot{Z}=Z$.\\
%%%%%%%%%%%%%%%%%%%%%%%%

\textbf{Case 29}: The system reduces to
\begin{equation*}
	\dot{x}=x(1+ax+ey+cz), \qquad \dot{y}=y(-3+ax+ey-kz), \qquad \dot{z}=z(1+kz).
\end{equation*}
The system has an invariant algebraic curve $l=1+kz=0$ with cofactor $kz$. Then we find a
first integral $\Phi=x^{-1}yz^{4} l^{(\frac{c}{k}-3)} $ with an inverse Jacobi multiplier $x y^{2}z^{4}
l^{-1}$. The other independent first integral is $\Psi=y^{-1}z^{-3}(1+O(x,y,z))$.
Therefore the two first integrals
in desired forms are $\phi_{1}=\Phi^{-3}\Psi^{-4}$ and $\phi_{2}=\Psi^{-1}$.\\
%%%%%%%%%%%%%%%%%%%%%%%%%%%%%%%%

\textbf{Case 30}:  The system takes the form
\begin{equation*}
	\dot{x}=x(1+ax+ey+cz), \qquad \dot{y}=y(-3+ax+ey-2kz), \quad \dot{z}=z(1+kz).
\end{equation*}
This case has an invariant algebraic curve $l=1+kz=0$ with cofactor $kz$. We find the first integral
$\Phi=x^{-1}yz^{4} l^{(\frac{c}{k}-2)} $ and the inverse Jacobi multiplier is $x y^{2}z^{4}$.
Theorem \ref{w} guarantees that the existence of
the other independent first integral is $\Psi=y^{-1}z^{-3}$. Thus, the two first integrals in desired form
are $\phi_{1}=\Phi^{-3}\Psi^{-4}$ and $\phi_{2}=\Psi^{-1}$. Since  $ \xi=z(1+kz)^{-1}$ satisfies $\dot{\xi}= \xi$, then the system must also be linearizable by Theorem \ref{w1}.\\
 
%%%%%%%%%%%%%%%%%%%%%%%%%%%%%%%
\textbf{Case 31}: The system becomes  
\begin{equation*}
	\dot{x}=x(1+ax-hy+3kz),\qquad  \dot{y}=y(-3-ax-kz),\qquad  \dot{z}=z(1+3ax+hy+kz).
\end{equation*}
We use nonstandard approach again as in Case 17. Assuming $a \simeq 0$ implies that
$d \simeq 0$ and $g \simeq 0$, then the system becomes
\begin{equation*}
	\dot{x} \simeq x(1-hy+3kz),\qquad  \dot{y} \simeq y(-3-kz),\qquad \dot{z} \simeq z(1+hy+kz).
\end{equation*}
The change of coordinates $(X,Y,Z)=(xy^{3},y,yz)$ transforms the system to 
\begin{equation*}
	\dot{X} \simeq X(-8-hY),\qquad  \dot{Y} \simeq -3 Y-kZ,\qquad  \dot{Z} \simeq Z(-2+hY).
\end{equation*}
The critical point at the origin of this system is in the Poincar\'e domain and hence is linearizable
via an analytic change of coordinates which can be chosen as $ (\hat{X},\hat{Y},\hat{Z}) \simeq (X(1+O(1)),Y(1+O(1)),Z(1+O(1)))$.
The two first integrals are  $\hat{\phi}=\hat{X}^{-3} \hat{Y}^2$ and $\hat{\psi}= \hat{Y}^2 \hat{X}^{-3}$, then pulling back these two first integrals, we see that the original system has two independent first integrals of required form $\phi_{1} =\hat{\phi}^{-1} =x^{3} y(1+\cdots)$ and $ \phi_{2}=\hat\psi^{-1}=y z^3 (1+\cdots)$.\\	
	
%%%%%%%%%%%%%%%%%%%%%%%%%%%%%%%%%%%%
\textbf{Case 32}: 
	When $k=0$ we have two subcases

	i)  $c=0 , f=0 , b=h , d=-a $, 
	then the system takes the form 
	\begin{eqnarray*}
		\dot{x}=x(1+ax+hy),\qquad 
		\dot{y}=y(-3-ax+ey),\qquad
		\dot{z}=z(1+gx+hy).
	\end{eqnarray*}
	is the same proof as in Case 19.\\
ii) $d=0 ,a=0,f=0,h=b$
then the system in this case is  
		\begin{equation*}
		\dot{x}=x(1+by+cz),\quad \hspace{1cm} \dot{y}=y(-3+ ey),\quad \hspace{1cm} \dot{z}=z(1+gx+by)
	\end{equation*}
is the same proof of Case 7.\\ 
If $k \neq 0$ then the system takes 
	\begin{equation*}
	\dot{x}=x(1- dx+hy+cz),\qquad \dot{y}=y(-3+dx+ ey-kz),\qquad \dot{z}=z(1+\frac{(-2k+c)}{k} dx+hy+kz).
\end{equation*}
The substitution $ (X,y,Z)=(dx-kz, y, yz)$ the system reduces to
\begin{equation*}
	\dot{X}=X(1-X+hy), \qquad \dot{y}=y(-3+X+ ey), \qquad \dot{Z}=Z(-2+\frac{(c-k)}{k} dX+(e+h)y).
\end{equation*}
The first two equations give linearizable saddle by Theorem \ref{colin}, part {\it i} and clearly
$\dot{\hat{X}}=\hat{X}$ and $ \dot{Y}=-3 Y$.\\
The transformation $\hat{Z}= Z \exp(-\gamma(\hat{X},Y))$ will linearize the third equation.
It just remains to find $\gamma(\hat{X},Y)$ such that\\
\begin{equation}\label{36}
\dot{\gamma}= \frac{(c-k)}{k} dX (\hat{X},Y)+(e+h) y(\hat{X},Y).
\end{equation} 
To see this, let 
\begin{equation}\label{37}
\gamma(\hat{X},Y)= \sum\nolimits_{i+j>0} \alpha_{ij} \hat{X}^{i} Y^{j}.
\end{equation}
 Also assume $X(\hat{X},Y)=\sum\nolimits_{i+j>0} g_{ij} \hat{X}^{i} Y^{j}$ and $y(\hat{X},Y)=\sum\nolimits_{i+j>0} h_{ij} \hat{X}^{i} Y^{j}$. \\
Differentiating (\ref{37}) and equating with (\ref{36}), we get 
$$\sum\nolimits_{i+j>0} (i-3j)\alpha_{ij}X^{i} Y^{j}=\sum\nolimits_{i+j>0}(\frac{c-k}{k}) d g_{ij} \hat{X}^{i} Y^{j}+\sum\nolimits_{i+j>0} (e+h)h_{ij} \hat{X}^{i} Y^{j} ,$$
then we can choose  $\alpha_{ij}$=$\frac{(\frac{c-k}{k}) d g_{ij}+(e+h) h_{ij}}{i-3j}$ for $ i \neq 3j $ or otherwise $\alpha_{ij}=0$. Therefore, the system is linearizable.\\
%%%%%%%%%%%%%%%%%%%%%%%%%%%%%%%%

\textbf{Case 33}: The system takes the form  
	\begin{eqnarray*}
		\dot{x}=x(1-gx+by+kz),\qquad \dot{y}=y(-3-gx+ey-kz),\qquad \dot{z}=z(1+gx+hy+kz).
	\end{eqnarray*}
The change of coordinates $ (X,Y,Z)=(xy,y,yz) $ reduce the system into
	\begin{eqnarray*}
		\dot{X}=X(-2+(b+e)Y),\qquad \dot{y}=Y(-3+eY)-gX-kZ,\qquad \dot{Z}=Z(-2+(e+h)Y)
	\end{eqnarray*}
The critical point at the origin of system is in the Poincar\'e domain and hence is linearizable by the analytic change of variables which can be chosen to be of the form\\ $ (\hat{X},\hat{Y},\hat{Z})=(X(1+O(1)),Y-gX-kZ+O(2),Z(1+O(1))) $.\\
The two first integrals $ \hat{\phi}=\hat{X}^{-3} \hat{Y}^{2} $ and  $ \hat{\psi}=\hat{Y}^{2} \hat{Z}^{-3}$ pull back these two first integrals, we see that the initial system has two independent analytic first integrals of desired form $\phi_{1} =\hat{\phi}^{-1} =x^{3} y(1+\cdots)$ and $ \phi_{2}=\hat\psi^{-1}=y z^3 (1+\cdots)$.\\

%%%%%%%%%%%%%%%%%%%%%%%%%%%%%%%%%%%
\textbf{Case 34}: The system (\ref{2}) can be written as   
	\begin{equation*}
		\dot{x}=x(1+ax+cz),\qquad  \dot{y}=y(-3+dx+fz),\qquad  \dot{z}=z(1+gx+kz).
	\end{equation*}
The first and third equations are linearizable node, with the change of coordinates $X=x(1+O(x,z))$ and $ Z=z(1+O(x,z)$) such that $\dot{X}=X,\dot{Z}=Z$. The second equation can be linearized by the transformation $Y=y e^{-\eta}$ if $\eta$ can be selected so that
\begin{equation*}
\dot{{\eta}}(X,Z)= d x (X,Z)+fz(X,Z).
\end{equation*}
Let $ \eta(x,Z)=\sum\nolimits_{i,j \ge 0} a_{ij} X^iZ^j$. Then 
$$\dot{{\eta}}=\sum\nolimits_{i,j \ge 0}(i+j) a_{ij} X^iZ^j=dx(X,Z)+fz(X,Z)=\sum\nolimits_{i,j \ge 0} b_{ij} X^iZ^j.$$Setting $a_{ij}=\frac{b_{ij}} {i+j}$  for $i+j > 0 $ which yields the convergence of $\eta$. \\
%%%%%%%%%%%%%%%%%%%%%%%%%%%%%%%%

\textbf{Case 35}: When $a \neq 0$, then the system can be written as  
	\begin{equation*}
		\dot{x}=x(1+ax+(2k-\frac{gk}{a})z), \qquad \dot{y}=y(-3-2ax+ey-2kz), \qquad \dot{z}=z(1+gx+kz).
	\end{equation*}
We can see that the first and third equations of the system are lineaizable node and the change of variable  $Y=y l_{1}^3 l_{2}^{-1}$ will linearize the second equation, where $\l_{1} = 1+ax+kz=0 $ and $\l_{2} = 1+ax-\frac{1}{3}ey+kz=0 $.\\
When $a = 0$, we have two sub-cases.
\begin{enumerate}[i)]
\item $d=g=b=h=f+2k=0$, then the system becomes 
	\begin{equation*}
	\dot{x}=x(1+cz), \qquad \dot{y}=y(-3+ey-2kz), \qquad \dot{z}=z(1+kz).
\end{equation*}
	The transformation $(X,Y,Z)=(xl_{1}^{-\frac{c}{k}},y(l_{1}l_{2})^{-1},zl_{1}^{-1})$ will linearize the system where $l_{1}=1+kz$ and  $l_{2}=1-\frac{e}{3}y+kz$ .\\
\item $d=k=b=h=f=0$, then the system takes the form 
\begin{equation*}
	\dot{x}=x(1+cz), \qquad \dot{y}=y(-3+ey), \qquad \dot{z}=z(1+gx).
\end{equation*}
The system has an invariant algebraic curve $l=1-\frac{e}{3}y=0$ with cofactor $ey$. 
The second equation linearized by the substitution $Y= y l^{-1}$
 gives $\dot{Y}=-3Y$. The remained equations are in the variables $x$ and $z$ which give linearizable node and by the change of coordinates $X=x(1+O(x,z))$ and $Z=z(1+O(x,z))$ will linearize so that
$\dot{X}=X$, $\dot{Z}=Z$.
\end{enumerate}
%%%%%%%%%%%%%%%%%%%%%%%%%%%%%%%%%

\textbf{Case 36}: The system can be written as 
\begin{equation*}
	\dot{x}=x(1+ax+3kz), \qquad \dot{y}=y(-3-2ax+ey-2kz), \qquad \dot{z}=z(1+3ax+kz).
\end{equation*}
This case have invariant algebraic surfaces
  $\l_{1}=1-\frac{1}{3}ey+\frac{1}{3}eaxy+\frac{1}{3}ekyz-\frac{1}{3}ea^{2}yx^{2}+\frac{2}{3}kea xyz-\frac{1}{3}ek^{2}yz^{2}=0 $  and  $\l_{2} = 1+2ax2kz-2kaxz+(ax)^{2}+(kz)^{2}=0$ with cofactors $ey $ and $2(ax+kz)$, respectively.
The transformation $Y=y l_{1}^{-1} l_{2} $	 linearizes the second equation. Note also that the first and third equations of the system are lineaizable node.\\

%%%%%%%%%%%%%%%%%%%%%%%%%%%
\textbf{Case 37}: The integrability conditions in this case takes the form
 \begin{equation*}
	\dot{x}=x(1+ax-5hy+3kz), \qquad \dot{y}=y(-3-3hy-3kz), \qquad \dot{z}=z(1+hy+kz).
\end{equation*}
 The system has an invariant algebraic surface
	$ l=(1+hy+kz)^{2}+\frac{1}{2} ax(2-hy+kz)=0$ with cofactor $ax-6hy+2kz$, which gives two independent first integrals
	$\phi_{1}=x^{3}y l^{-3}$ and $\phi_{2}=yz^{3}$.\\
	
%%%%%%%%%%%%%%%%%%%%%%%%%%%%%%%%%%%
\textbf{Case 38}: The system has the form
	\begin{equation*}
		\dot{x}=x(1+ax-8hy+4kz),\qquad  \dot{y}=-3y(1+hy+kz),\qquad  \dot{z}=z(1+hy+kz).
	\end{equation*}
It is obvious we have a first integral $ \Phi= y z^{3} $.
We can also find an inverse Jacobi multiplier $x^{2} y^{1/3} z^{-2} l^{-2} $ which yields the second
first integral  $\Psi= x^{-1}  y^\frac{1}{3} z^{3} l^{-2}$.
We write the independent first integrals in required forms
$\phi_1=\Phi^{3} \Psi^{-3}=x^{3} y (1+\cdots)$  and $\phi_2=  y z^{3} $.\\
	
%%%%%%%%%%%%%%%%%%%%%%%%%%%%%%%%%%
\textbf{Case 39}: The system becomes 
	\begin{equation*}
		\dot{x}=x(1+ax-8hy+4kz), \qquad \dot{y}=y(-3+hy+kz), \qquad \dot{z}=z(1-2ax+hy+kz).
	\end{equation*}
In this case, the system has an invariant algebraic surface
$\l = (1+hy+kz)^{3}+akxz(3+ax-6hy+2kz)=0$ with cofactor $ -9hy+3kz $ which lets us deduce a first integral  $\Phi=x^{2} y z l^{-2}$ and an inverse Jacobi multiplier $x^{2} y^{4/3} z l^{-2/3}$.
We can obtain a second integral $\Psi=x^{-1} y^{\frac{-1}{3}}(1+\cdots)$ by Theorem \ref{w}.
We write these first integrals in desired forms $\phi_1=\Psi^{-3} =x^{3} y (1+cdots)$  and
$\phi_2=(\Psi)^{6}\Phi^{3}=  y z^{3} (1+\cdots)$.\\

%%%%%%%%%%%%%%%%%%%%%%%%%%%%%
\textbf{Case 40}: The system reduces to
\begin{equation*}
	\dot{x}=x(1-gx-5hy+3kz), \qquad \dot{y}=y(-3+gx-3hy-3kz), \qquad \dot{z}=z(1+gx+hy+kz).
\end{equation*}
Invariant algebraic surfaces in this case are $l_{1}=(1+hy+kz)+gx( hy-kz)=0$ and
$l_{2}=(1-gx+hy+kz)^{2}+4ghxy=0$ with respective cofactors $-2(3hy-kz)$ and $-2(gx+3hy-kz)$
 yielding two independently first integrals as follows $\phi_1=x^{3} y l_{1}^{-2} l_{2}^{-1}$  and
$\phi_2= y z^{3} l_{1}^{-2} l_{2}^{2}$.\\
%%%%%%%%%%%%%%%%%%%%%%%%%%%%%%%%%%%%%%%

\textbf{Case 41}:	If $ e\neq 0$, then we have two subcases.
\begin{enumerate}[1.]
\item The integral conditions appears as
	 $$k(b+e-h)-ce=a(3b-2e)-de=k(2e-3h)+fe=a(b-e-h)+ge=0.$$
The system has an invariant algebraic surface  $l=1+ax-\frac{1}{3}ey+kz=0$ with cofactor
$ax+ey+kz$ and produces two first integrals  $\phi_1=x^{3} y l^{-1-\frac{3b}{e}}$  and
$\phi_2= y z^{3} l^{-1+\frac{3h}{e}}$.\\
	
\item With the conditions $a=g=-\frac{d}{3}, b=h=-\frac{e}{3} $ and $c=k=-\frac{f}{3}$, the system (\ref{2}) reduces to 
		\begin{equation*}
		\dot{x}=x(1-\frac{d}{3}x-\frac{e}{3}y-\frac{f}{3}z), \qquad \dot{y}=y(-3+dx+ey+fz), \qquad \dot{z}=z(1-\frac{d}{3}x-\frac{e}{3}y-\frac{f}{3}z)
	\end{equation*} 
Therefore, the system has two independent first integrals of the form 
$\phi_1=x^{3} y$  and $\phi_2 = y z^{3}$.\\
\end{enumerate}

We now consider the case when $e=0$. Again we have some subcases.
\begin{enumerate}[i.]
\item $h=k=0$.\\
	     	 If $b=0$, $a=0$ and $d=\frac{-g(f+3c)}{c}$, the system has exponential factor $E_1=\exp(-gx+cz)$ with cofactor $-gx+cz$, which yields two first integrals 
	  $\phi_1=x^{3} y E_1^{-3-\frac{f}{c}}$  and $\phi_2= y z^{3} E_1^{\frac{-f}{c}}$.\\
	  	  
	 	 If $b\neq 0$ , $a=0$ $f=0$ and $d=-3g$, then the system admits an exponential factor $E_2=\exp(3gx+by-3cz)$ with cofactor $3(gx-by-cz)$, and we construct two first integrals 
	  $\phi_1=x^{3} y E_2 $ and $\phi_2= y z^{3} $.\\
	  
\item If $h\neq 0$, $k=0$, $a=0$ ,$b= \frac{h(f+3c)}{f}$ and $d=- \frac{g(f+3c)}{c}$, we have an exponential factor $E_3=\exp(-fgx-hey+fcz)$ with cofactor $-fgx+3hey+fcz$, with two first integrals 
	  $\phi_1=x^{3} y E_3^{\frac{-(3c+f)}{fc}} $ and $\phi_2= y z^{3} E_3^{\frac{-1}{c}} $.\\
	  
\item If $k\neq 0$, $h=0$ $b=0$, $c=\frac{k(2a-g)}{a}$and $f=- \frac{k(3a-d-3g)}{a}$, the system has an algebraic surface $\l=1+ax+kz=0$ with cofactor $ax+kz$, with two first integrals 
	  $\phi_1=x^{3} y l^{-\frac{(3a+d)}{a}} $ and $\phi_2= y z^{3} l^{-\frac{(d+3g)}{a}} $.
\end{enumerate}
%%%%%%%%%%%%%%%%%%%%%%%%%%%%%%%

 \textbf{Case 41.1}: The system is of the form 
	\begin{equation*}
	\dot{x}=x(1+dx+ey), \qquad  \dot{y}=y(-3+dx+ey), \qquad  \dot{z}=z(1+gx+\frac{e g}{d}y).
\end{equation*}
We have two subcases depending on the parameter $d$.
\begin{enumerate}[1.]
\item  When $d\neq 0$, the change of variables $X=x  l^{-1}$, $Y=y  l^{-1}$ and $Z=z  l^{-\frac{g}{d}}$
where  $l=1+ d x -\frac{e}{3} y=0$ linearizes the system.

\item When $d= 0$, $e \neq 0$ and $g=0$, the change of variables $X=x  l^{-1}$, $Y=y  l^{-1}$ and
$Z=z  l^{-\frac{h}{e}}$ linearizes the system.\\
 Now, if $ e = 0 $ and $g \neq 0$, the change of variables $X=x$, $Y=y$ and $Z=z  E$
 where  $E=\exp(gx-\frac{h}{3}y)$ linearizes the system.\\
 \end{enumerate}
  %%%%%%%%%%%%%%%%%%%%
  \textbf{Case 41.2}: The system in this case has the form 
 \begin{equation*}
 	\dot{x}=x(1+gx+hy+kz), \qquad	\dot{y}=y(-3+gx+hy+kz), \qquad \dot{z}=z(1+gx+hy+kz).
 \end{equation*} 
 The system will be linearized by the change of coordinates  $ (X,Y,Z)=(x  l^{-1},y l^{-1},z l^{-1})$ where   $l$    is an invariant algebraic surface of the form $l=1+g x+h y+k z=0$ with cofactor
 $ g x+h y +k z $. \\
		\end{proof}
%%%%%%%%%%%%%%%%%%%%%%%%%%%%%%%%%%%%%%%%%%%%%%%%%%%%%%%%%%%%%%%%%%%%%%%%%%%%%%%%%%%%%%%%%%%%%%%%%%%%
\section{Nonstandard results}

 This section motivated two different approaches to construct two independent first integrals.
 The first motivation is infinitely close to standard first integrals by using nonstandard
 perturbation theory of differential equations. The other motivation uses the unperturbation
 nonstandard transformation to  find infinitesimal, limited or unlimited first integrals.\\
 We study two types of system (\ref{ch}) depends on eigenvalues which are $\delta =(1+\varepsilon)$, $\varepsilon$ is infinitesimal, in these cases, either $\delta$ is appreciable or $\delta$ is itself infinitesimal.\\

Now, the change of variables
$ x_{1}=\delta^{-1} x$, $ y_{1}=\delta^{-1} y$, $ z_{1}=\delta^{-1} z$ and  $ \tau =\delta t$ transform the system (\ref{ch}) with respect to $\tau$ to 
\begin{equation}\label{38}
\begin{aligned}
	{x'_{1}}&=&x_{1}(1+ax_{1}+by_{1}+cz_{1}), \\
	{y'_{1}}&=&y_{1}(-3+dx_{1}+ey_{1}+fz_{1}), \\
	{z'_{1}}&=&z_{1}(1+gx_{1}+hy_{1}+kz_{1}),
	\end{aligned}
\end{equation}
where $'$ stands for $\frac{d}{d\tau}$ and in general the variables and parameters are hyperreals.
%%%%%%%%%%%%%%%%%%%%%%%%%%%%%%%
\subsection{Perturbation Case}
In this subsection, the study of the system with  $\delta =(1+\varepsilon)$ satisfies the perturbation technique to find first integrals and linearization of differential system, besides that, the  $\delta$ is infinitesimal type of system and it does not satisfy the perturbation technique.\\ 
Where $\delta =(1+\varepsilon)$  the system becomes 
 	\begin{eqnarray}\label{39}
 	F=\left\{\begin{array}{rc} 
 	\dot{x}&=F_{1}=x((1+\varepsilon)+ax+by+cz),\\ \dot{y}&=F_{2}=y(-3 (1+\varepsilon)+dx+ey+fz),\\
 	\dot{z}&=F_{3}=z((1+\varepsilon)+gx+hy+kz), \end{array}\right.
 		\end{eqnarray}
	where $F=(F_{1},F_{2},F_{3})$, where variables are near standard numbers and parameters are standard.\\
%%%%%%%%%%%%%%%%
Let $\mathbb{R}^3$ be a standard topological space. A point $p=(x,y,z) \in \mathbb{R}^3$ is said to be infinitely close to a standard point $p_{0}=(x_{0},y_{0},z_{0}) \in \mathbb{R}^3$ which is denoted by $p\simeq p_{0} $, if $p$ is in any standard neighborhood of $p_{0}$. Let $D_{0}$ be a subset of $\mathbb{R}^3$. A point $p \in \mathbb{R}^3$  is said to be nearstandard in $D_{0}$ if there is a standard $p_{0}$ $\in$$D_{0}$ such that $p\simeq p_{0} $ . Let us denote by
\begin{align*} 
  {}^{N S}D_{0} =\{{p \in\mathbb{R}^3:\exists^{st}p_{0} \in D_{0}, ,p\simeq p_{0}}\},
\end{align*}
 the external-set of nearstandard points in $D_{0}$.\\
%%%%%%%%%%%%%%%%%%%%%
The reformulation of Theorem \ref{int} to standard case by using the definition of standard with $\varepsilon=0$ and its nonstandard formulation proof is the same as the classical proof.

 \begin{proposition}\label{I}
Let $ F : D \to \mathbb{R}^3$ and $ F_{0} : D_{0} \to \mathbb{R}^3 $ be vector fields, $p=(x,y,z)\in D$
and $p_{0}=(x_{0},y_{0},z_{0})\in D_{0}$ and  $p_{0}$ is standard. Then vector field $F$ is
perturbation to the vector field $F_{0}$ and system (\ref{39}) has the same integrability
conditions of system (\ref{2}), if  $p\simeq p_{0}$ and $F\simeq F_{0} $, for topology of
uniform convergence on compacta.
 \end{proposition}
 \begin{proof} 
According to Heine Borel's theorem, $\mathbb{R}^{3}$ is locally compact (any compact set is locally
compact), we assume $p=(x,y,z) \in {}^{NS}D_{0}$, then there exists a standard $p_{0}=(x_{0},y_{0},z_{0}) \in D_{0}$ such that $p \simeq p_{0} $.\\
To prove $F$ is infinitely close to $F_0$, that means
\begin{align*} 
 F\simeq F_{0}.
\end{align*}
Since $\mathbb{R}^{3}$ is locally compact, there exists a standard compact neighborhood $ K $ of $p_{0}$ such that $p\simeq p_{0}$.
Hence $ p \in K \subset D$  and $p=(x,y,z)$ belongs to standard set $D_{0}$.\\
 Now, $F_{1}(p)=x (1+\varepsilon+a x+b y+c z)$ and $F_{01}(p)=x (1+a x+b y+c z)$, 
so 
\[
F_{1}(p)-F_{01}(p)=x(1+\varepsilon+a x+b y+c z)-x(1+a x+b y+c z)=x \varepsilon.
\]
 Then
\[
 F_{1}(p) \simeq  F_{01}(p).
\]
 We repeat the processes for the two other $ F_{2}(p) \simeq F_{02}(p)$ and $ F_{3}(p) \simeq F_{03}(p)$.\\
Thus, $F$ is infinitely close to $F_{0}$, then $F$ is perturbation to $F_{0}$.\\
Using the same transformation as in the system (\ref{38}) into the system (\ref{39}), we get the same integrability conditions as in Theorem \ref{int}. \\
 \end{proof}
%%%%%%%%%%%%%%%%%%%%%%%%%%%%%%%%%%
 \begin{proposition}\label{ss}
 Let a vector field $ F : D \to \mathbb{R}^3$ is a perturbation of a standard vector field $ F_{0} : D_{0} \to \mathbb{R}^3$ and system {(\ref{39})} has two first integrals of the form $H_{1}(x,y,z)=x^{3} y (1+O(x,y,z))$ and $H_{2}(x,y,z)=y z^{3} (1+O(x,y,z))$ where $(x,y,z) \in {}^{NS}D_{0}$, then  $H_{1}(x,y,z) \simeq \phi_{1}(x,y,z)$ and  $H_{2}(x,y,z) \simeq \phi_{2}(x,y,z)$.\\
 \end{proposition}
 \begin{proof} By Proposition \ref{sari}, the vector field $F=(F_{1},F_{2},F_{3}) $ in {(\ref{39})} is a perturbation to the standard vector field $F_{0}=(F_{01},F_{02},F_{03}) $ in system (\ref{2}).

Since $(x,y,z) \in {}^{NS}D_{0}$, we can set $x=x_{0}+\varepsilon_{1} $, $y=y_{0}+\varepsilon_{2} $ and $z=z_{0}+\varepsilon_{3} $, where $\varepsilon_{1},\varepsilon_{2},\varepsilon_{3}$ are infinitesimal and $(x_{0},y_{0},z_{0}) \in D_{0}$.\\
\begin{align}\label{H}
  H_{1}(x_{0}+\varepsilon_{1},y_{0}+\varepsilon_{2},z_{0}+\varepsilon_{3})-\phi_{1}(x_{0},y_{0},z_{0}) 
   &=(x_{0}+\varepsilon_{1})^{3} (y_{0}+\varepsilon_{2}) (1+O(x_{0}+\varepsilon_{1},y_{0}+\varepsilon_{2},z_{0}\nonumber \\
   &+\varepsilon_{3}))-x_{0}^{3} y_{0} (1+O(x_{0},y_{0},z_{0})).
  \end{align}
 Performing some calculations, we can simplify equation \eqref{H} to 
   $$x_{0}^{3} y_{0} (1+O(x_{0},y_{0},z_{0})+\tau_{1})-x_{0}^{3} y_{0} (1+O(x_{0},y_{0},z_{0}))=x_{0}^{3} y_{0}\tau_{1},$$
  where $\tau_{1}$ is infinitesimal 
 and this implies
$ H_{1} \simeq  \phi_{1}$ for all near standard point of $D_{0}$\\
In similar way we can obtain  $ H_{2} \simeq \phi_{2}$.\\
 \end{proof}
%%%%%%%%%%%%%%%%%%%%%%%%%
 \begin{proposition}
 If the vector field  $ E_{1}: D \to \mathbb{R}^3$ is a perturbation of a standard vector field $ E_{0}: D_{0} \to \mathbb{R}^3$, with $p \in {}^{NS}D_{0}$ and $p_{0} \in D_{0}$, then the linearizablity of the two vector fields (\ref{39}) and (\ref{2})  are $E_{1}\simeq E_{0}$ and $ p \simeq p_{0} $, where 
 
 \begin{equation*}\label{40}
 E_{1}=\left\{\begin{array}{rcc} 
 \dot{x}&=E_{11}=(1+\varepsilon)x,\\ \dot{y}&=E_{12}=-3(1+\varepsilon)y),\\
 \dot{z}&=E_{13}=(1+\varepsilon)z, \end{array}\right.
 \end{equation*} 
 and
  \begin{equation*}
 E_{0}=\left\{\begin{array}{rccc} 
 \dot{x}&=E_{01}=x,\\ \dot{y}&=E_{02}=-3y,\\
 \dot{z}&=E_{03}=z, \end{array}\right.
 \end{equation*}
 where $E_{1}=({E_{11},E_{12},E_{13}})$ and $E_{0}=({E_{01},E_{02},E_{03}})$
 \end{proposition}
\begin{proof}:
	Using Proposition \ref{ss}, the vector field $E$ is a perturbation to the standard vector field  $E_{0}$. If we let $p=(x_{0}+\varepsilon_{1},y_{0}+\varepsilon_{2} ,z_{0}+\varepsilon_{3}) $ and
	$p_{0}=(x_{0},y_{0},z_{0})$ where $\varepsilon_{1},\varepsilon_{2}$ and $\varepsilon_{3}$
are infinitesimals, then directly from the vector fields $E_{1}$ and $E_{0}$, we conclude that
$p \simeq p_{0}$ and  $E_{1}\simeq E_{0}.$ 
\end{proof}
%%%%%%%%%%%%%%%%%%%%%%%%%%%%%%%%%%%%%%
 \begin{remark} In system \eqref{ch}, when $\delta$ is infinitesimal, then the Proposition \ref{sari} for perturbation does not work because the difference of the two vector fields is appreciable not infinitesimal.
\end{remark}
 %%%%%%%%%%%%%%%%%%%%%%%%%%%%%%%%%%
  \subsection{Unperturbation Case}
  In this section we analyze the two independent first integrals in the perspective view of
  nonstandard analysis. 
   \begin{proposition}\label{A}
   If the conditions of Theorem \ref{int} are satisfied, then the  system \eqref{ch} with $ \delta=1+\varepsilon $ has two independent first integrals of the form $\psi_{1}= {(1+\varepsilon)}^{-4}x^{3} y (1+O({(1+\varepsilon)}^{-1}x$,\\${(1+\varepsilon)}^{-1}y,{(1+\varepsilon)}^{-1}z))$ and $\psi_{2}={(1+\varepsilon)}^{-4} y z^{3} (1+O({(1+\varepsilon)}^{-1}x,{(1+\varepsilon)}^{-1}y,{(1+\varepsilon)}^{-1}z))$ where $(x,y,z) \in \textbf{R}^3$ where parameters are standard, and the two first integrals are limited or unlimited.
   \end{proposition}
 \begin{proof}
 The transformation of  \eqref{38}, reduces the system {\eqref{39}} with $ \delta=1+\varepsilon $  to
 	 \begin{eqnarray*}
 	 	{x'_{1}}&=&x_{1}(1+ax_{1}+by_{1}+cz_{1}),\\
 	 	{y'_{1}}&=&y_{1}(-3+dx_{1}+ey_{1}+fz_{1}),\\
 	 	{z'_{1}}&=&z_{1}(1+gx_{1}+hy_{1}+kz_{1}).
 	 \end{eqnarray*}  
 Using the arguments similar in Theorem \ref{int}, the system has two independent first integrals of
 the form  $\phi_{1}=x_{1}^{3} y_{1} (1+O(x_{1},y_{1},z_{1}))=x_{1}^{3} y_{1} p_{1}$ and
 $\phi_{2}= y_{1} z_{1}^{3} (1+O(x_{1},y_{1},z_{1}))=y_{1} z_{1}^{3} p_{2}$.\\
 then the first integrals for the original system are 
 $$\psi_{1}={(1+\varepsilon)}^{-4} x^{3}y(1+O({(1+\varepsilon)}^{-1}x,{(1+\varepsilon)}^{-1}y,{(1+\varepsilon)}^{-1}z))={(1+\varepsilon)}^{-4}x^{3} y p_{1}$$
 and 
 $$\psi_{2}={(1+\varepsilon)}^{-4} y z^{3} (1+O({(1+\varepsilon)}^{-1}x,{(1+\varepsilon)}^{-1}y,{(1+\varepsilon)}^{-1}z))={(1+\varepsilon)}^{-4}y z^{3} p_{2}$$
 Since the variables are nonstandard and parameters are standard, then we have different
 types of first integrals.
  \begin{enumerate}[I.]
  \item The two first integrals $\psi_{1},\psi_{2}$ are appreciable .\\ 
  \textbf{Case 1:} If $x,y,z \in  \mbox{gal}(0)$, then the two first integrals are obtained directly from the substitution of   $x,y,z$ in  $\psi_{1}$ and $\psi_{2}$.\\\\
 \textbf{Case 2:} If $y \in \mbox{m}(\omega) $, wher $\omega$ is unlimited and $(1+O({(1+\varepsilon)}^{-1}x,{(1+\varepsilon)}^{-1}y,{(1+\varepsilon)}^{-1}z))$ contain $y$, then by
assuming $y (1+O({(1+\varepsilon)}^{-1}x,{(1+\varepsilon)}^{-1}y,{(1+\varepsilon)}^{-1}z))=y_{1}  $ and since the product of two unlimited is unlimited therefore $y_{1} $ is unlimited. Let $x=\frac{\alpha}{\sqrt[3]{y_{1}}}$ and $z=\frac{\beta}{\sqrt[3]{y_{2}}}$ where $\alpha, \beta \in \mbox{gal}(0)$. Then we deduce that the  two first integrals are appreciable.\\
   
 \item The two first integrals $\psi_{1}, \psi_{2}$ are infinitesimal\\
          \textbf{Case 1:} If $x,y,z \in \mbox{m}(0)$, then the infinitesimal of the two first integrals are directly obtained by the substitution of variables in $\psi_{1}$, $\psi_{2} $.\\\\
          \textbf{Case 2:} If $y \in \mbox{m}(0)$ and  $x,z \in \mbox{gal}(0)$, the multiplication of $x^{3}$ and $ (1+O({(1+\varepsilon)}^{-1}x,\\{(1+\varepsilon)}^{-1}y,{(1+\varepsilon)}^{-1}z))$ stay limited. When we multiply the results by $y$, we see that the two first integrals are infinitesimal. \\\\
     \textbf{Case 3:} If $y \in \mbox{m}(\omega)$ and $(1+O(x,y,z))$ contain $y$, then
suppose $y (1+O({(1+\varepsilon)}^{-1}x,\\{(1+\varepsilon)}^{-1}y,{(1+\varepsilon)}^{-1}z))=y_{1} $ and the product of two unlimited is unlimited, hence
 $y_{1} $ is unlimited. If we take $x=\frac{\alpha}{\sqrt[3]{y_{1}^{2}}}$ and $z=\frac{\beta}
 {\sqrt[3]{y_{2}^{2}}}$, where $\alpha, \beta \in \mbox{gal}(0)$, then the two first integrals are infinitesimal.\\
 
 \item The two first integrals $\psi_{1},\psi_{2}$ are unlimited\\
      If $ y \in \mbox{m}(\omega) $ and $(1+O({(1+\varepsilon)}^{-1}x,{(1+\varepsilon)}^{-1}y,{(1+\varepsilon)}^{-1}z))$ contain $y$, then suppose $y (1+O({(1+\varepsilon)}^{-1}x,{(1+\varepsilon)}^{-1}y,{(1+\varepsilon)}^{-1}z))=y_{1}$ and the product of two unlimited is unlimited, therefore $y_{1} $ is unlimited. If we take $x=\frac{\alpha}{\sqrt[6]{y_{1}}}$ and $z=\frac{\beta}{\sqrt[6]{y_{2}}}$ where  $\alpha, \beta \in \mbox{gal}(0)$, thus the results stay unlimited.\\
 \end{enumerate}
   \end{proof}
%%%%%%%%%%%%%%%%%%%%%%%%%%%%%%%%%%%%%%%%%%%%%555
 \begin{proposition}\label{B}
 		Let the system (\ref{ch}) with $ \delta=1+\varepsilon $ and all parameters are standard except $a$. If the conditions of Theorem \ref{int} are satisfied, then it has two independent first integrals of the form $\psi_{1}={(1+\varepsilon)}^{-4}x^{3} y (1+O({(1+\varepsilon)}^{-1}x,{(1+\varepsilon)}^{-1}y,{(1+\varepsilon)}^{-1}z))$ and $\psi_{2}= {(1+\varepsilon)}^{-4}y z^{3} (1+O({(1+\varepsilon)}^{-1}x,{(1+\varepsilon)}^{-1}y,{(1+\varepsilon)}^{-1}z))$,
which are either limited or unlimited.
 \end{proposition}	
 \begin{proof}
From the proof of the first part of Proposition \ref{A} with
  $ \delta=1+\varepsilon $, the system has two independent first integrals of the form  $\phi_{1}=x_{1}^{3} y_{1} (1+O(x_{1},y_{1},z_{1}))$ and $\phi_{2}= y_{1} z_{1}^{3} (1+O(x_{1},y_{1},z_{1}))$.
 Then the two first integrals of system(\ref{ch}) are
 $$\psi_{1}={(1+\varepsilon)}^{-4}x^{3} y (1+O({(1+\varepsilon)}^{-1}x,{(1+\varepsilon)}^{-1}y,{(1+\varepsilon)}^{-1}z))$$
 and
 $$\psi_{2}= {(1+\varepsilon)}^{-4}y z^{3} (1+O({(1+\varepsilon)}^{-1}x,{(1+\varepsilon)}^{-1}y,{(1+\varepsilon)}^{-1}z)).$$
The condition will be on the parameter $a$  and other parameters are standard. Let $G_{1}(G_{2}) $ be the set of coefficients $x,y,z$ in $O({(1+\varepsilon)}^{-1}x,{(1+\varepsilon)}^{-1}y,{(1+\varepsilon)}^{-1}z)$. We have different types of first integrals which are 
\begin{enumerate}[I.]
\item The two first integrals $\psi_{1},\psi_{2} $ are appreciable. \\
	\textbf{Case 1:} If $x(z)$, $y$ and $a$ are appreciable, then we have two appreciable first integrals directly from the substitution of variables in the first integrals. \\
	\textbf{Case 2:} If  $x(z)$ and $y$ are appreciable, $a$ is unlimited and $a^{-1} \in G_{1}(G_{2}) $, then the term $(1+O({(1+\varepsilon)}^{-1}x,{(1+\varepsilon)}^{-1}y,{(1+\varepsilon)}^{-1}z))$ stays appreciable and then is multiplied by $x^{3} y$. Therefore, the result is unchanged. \\
	\textbf{Case 3:} If  $x(z)$ and $y $ are appreciable when  $a \not\in G_{1}(G_{2}) $, the result stays appreciable. \\
	\textbf{Case 4:} If  $x(z)$ and $y$ are limited with  $ G_{1}(G_{2})=$\{0\}$ $, then obviously the two first integrals are appreciable since variables are appreciable. \\
	%%%%%%%%%%%%%%%%%%%
	
\item The two first integrals $\psi_{1},\psi_{2} $ are infinitesimal. \\
\textbf{Case 1:} If $x(z)$ and $y $ are infinitesimal, then the term $O({(1+\varepsilon)}^{-1}x,{(1+\varepsilon)}^{-1}y,{(1+\varepsilon)}^{-1}z)$ is
infinitesimal.  Therefore, $(1+O({(1+\varepsilon)}^{-1}x,{(1+\varepsilon)}^{-1}y,{(1+\varepsilon)}^{-1}z))$ is appreciable. When we multiply this appreciable 
term to  ${x}^3 ({z}^3)y $ which is infinitesimal, then the result is infinitesimal.\\
\textbf{Case 2:} If  $x(z)$ and $y $ are infinitesimal with  $a \not\in G_{1}(G_{2})$, and since the
coefficients of $O({(1+\varepsilon)}^{-1}x,{(1+\varepsilon)}^{-1}y,{(1+\varepsilon)}^{-1}z)$ does not contain $a$, then as the Case 1 (II) above it has the
same illustration. \\
\textbf{Case 3:} Let  $x(z)$ and $y $ are infinitesimal or   $x(z) $ is appreciable
(infinitesimal) and $y$  is infinitesimal (appreciable) with $ G_{1}(G_{2})=$\{0\}$ $. Since the
coefficients of $O({(1+\varepsilon)}^{-1}x,\\{(1+\varepsilon)}^{-1}y,{(1+\varepsilon)}^{-1}z))$ are zero, then the two first integrals are
$\psi_{1}={(1+\varepsilon)}^{-4}{x}^3 y $ and $\psi_{2}= {(1+\varepsilon)}^{-4}y{z}^3 $. When the variables are
infinitesimal, it is also true for $\psi_{1}$ and $\psi_{2}$ or when $x(z)$ is appreciable and
$y$ is infinitesimal, the result stays infinitesimal. \\
\textbf{Case 4:} Let  $x(z)$ and $y $ are infinitesimal or   $x(z) $ is infinitesimal and $y$  is appreciable with $ G_{1}(G_{2})=$\{0\}$ $. Since the coefficients of   $O({(1+\varepsilon)}^{-1}x,{(1+\varepsilon)}^{-1}y,{(1+\varepsilon)}^{-1}z))$ are zero, then the two first integrals are
$\psi_{1}={(1+\varepsilon)}^{-4}{x}^3 y $ and $\psi_{2}= {(1+\varepsilon)}^{-4}y{z}^3 $. When the variables are
infinitesimal, it is also true for $\psi_{1}$ and $\psi_{2}$ or when $x(z)$ is infinitesimal and
$y$ is appreciable , the result stays infinitesimal. \\
%%%%%%%%%%%%%

\item The two first integrals $\psi_{1},\psi_{2} $ are unlimited. \\

\textbf{Case 1:} If $x(z)$, $y$ and $a$ are unlimited, then it clearly obtains two unlimited first integrals. \\
\textbf{Case 2:} If  $x(z),y $ are appreciable, $a$ is unlimited and $a \in G_{1}(G_{2}) $, then  $O({(1+\varepsilon)}^{-1}x,\\{(1+\varepsilon)}^{-1}y,{(1+\varepsilon)}^{-1}z))$ is unlimited and when we multiply it by any appreciable the result is unlimited. \\
\textbf{Case 3:} If $ y $ is unlimited with $ G_{1}(G_{2})=\{0\}$, then the two first integrals become  $\psi_{1}={(1+\varepsilon)}^{-4}{x}^3 y $ and  $\psi_{2}={(1+\varepsilon)}^{-4}{ yz}^3 $, with  $x(z)=\frac{\alpha}{\sqrt[6]{y}}$ $(\frac{\beta}{\sqrt[6]{y}})$ where  $\alpha, \beta \in$ gal(0), and hence the two first integrals stay unlimited.
\end{enumerate}
	\end{proof}
%%%%%%%%%%%%%%%%%%%%%%%%%%%%%%%%%%%%%%%%%%%%%%%%%%%%%%
 \begin{proposition}\label{C}
Consider the system (\ref{ch})  with  $ \delta  $ is infinitesimal. If the conditions of Theorem \ref{int} are satisfied, then the system (\ref{ch}) has two independent first integrals of the form $\psi_{1}= {\delta}^{-4} x^{3} y (1+O({\delta}^{-1}x,{\delta}^{-1}y,{\delta}^{-1}z))$ and $\psi_{2}= {\delta}^{-4} y z^{3} (1+O({\delta}^{-1}x,{\delta}^{-1}y,{\delta}^{-1}z))$ which they are either limited or unlimited.
\end{proposition}
\begin{proof}
  	In a similar way as in the proof of Proposition \ref{A}, we have the two independent first integrals of the form $\psi_{1}= {\delta}^{-4} x^{3} y (1+O({\delta}^{-1}x,{\delta}^{-1}y,{\delta}^{-1}z))$ and $\psi_{2}= {\delta}^{-4} y z^{3} (1+O({\delta}^{-1}x,{\delta}^{-1}y,{\delta}^{-1}z))$ .  \\
 Now, if we suppose 	$ x=\delta^{n}(\alpha_{1}+\delta) $, $ y=\delta^{n}(\alpha_{2}+\delta) $ and  $ z=\delta^{n}(\alpha_{3}+\delta) $ where $\alpha_{1},\alpha_{2},\alpha_{3}$ are standard reals and $n$ is standard non-negative integer,
   then we have the following cases: \\
	\textbf{Case 1:} if $n=0$ implies $x,y,z$ $\in$ gal(0), then the two first integrals $\psi_{1},\psi_{2}$ are unlimited\\
	\textbf{Case 2:} If $n=1$ implies $x,y,z$ $\in$ m(0), then the two first integrals $\psi_{1},\psi_{2} $ are appreciable.\\
\textbf{Case 3:} If $1<n$ implies $x,y,z$  $\in$  $\delta$-micromonad(0), then the two first integrals $\psi_{1},\psi_{2}$  belong to $\delta$-micromonad(0). \\
\end{proof}
%%%%%%%%%%%%%%%%%%%%%%%
\begin{proposition}\label{D}
	Suppose the system (\ref{ch}) with $ \delta $ is infinitesimal and all parameters are standard reals except $a$. If the conditions of Theorem \ref{int} are satisfied then it has two independent first integrals of the form $\psi_{1}={\delta}^{-4} x^{3} y (1+O({\delta}^{-1}x,{\delta}^{-1}y,{\delta}^{-1}z))={\delta}^{-4} x^{3} y p_{1}$ and $\psi_{2}={\delta}^{-4} y z^{3} (1+O({\delta}^{-1}x,{\delta}^{-1}y,{\delta}^{-1}z))={\delta}^{-4} y z^{3} p_{2}$ and the two first integrals are either limited or unlimited.
\end{proposition}
\begin{proof}
		Using the first part of the proof of Proposition \ref{A}, then the two first integrals exist
		with $a$ either infinitesimal or limited, then the proof is similar to Proposition \ref{C}.
 Now for the remaining cases with unlimited parameter $a$,  assume $G_{1}(G_{2}) $ be the set of
 coefficients $x,y,z$ in $O({\delta}^{-1}x,{\delta}^{-1}y,{\delta}^{-1}z)$, then we have different types of first
 integrals.\\
 Suppose 	$ x=\delta^{n}(\alpha_{1}+\delta) $, $ y=\delta^{n}(\alpha_{2}+\delta) $ and  $ z=\delta^{n}(\alpha_{3}+\delta) $ for $\alpha_{1},\alpha_{2},\alpha_{3}$ are standard reals and $n$ is standard non-negative integer,
 then we have the following cases:\\
		\textbf{Case 1:} If $a$ is unlimited and $a^{-1} \in G_{1}(G_{2})$, then\\
			\begin{enumerate}[i.]
				\item	If $2 \leq n$ implies  $x,y,z \in \delta-\mbox{micromonad(0)}$, then 
 the two first integrals are belonging to  $\delta-\mbox{micromonad(0)}$.\\
	\item If $n=1$ implies that  $x,y,z \in \mbox{m}(0) $ and $p_{1}$ is appreciable, then the two first integrals are appreciable.\\
 \end{enumerate}
	\textbf{Case 2:} If $a$ is unlimited and $a \in G_{1}(G_{2})$.
	 Let $a={\delta}^{-1}$, we have the following cases.
	\begin{enumerate}[i.]
		\item If $1<n$ implies  $x,y,z \in \delta-\mbox{micromonad(0)}$, then the results depend on $p_{1}$ as follows.
		\begin{enumerate}
	        \item[1)]For $n=2$, then $p_{1}$ is appreciable and the two first integrals are infinitesimal.\\
     \item[2)]For $2<n$, then $p_{1}$ $\in$ $\delta-\mbox{micromonad(0)}$ and the two first integrals are belonging to $\delta-\mbox{micromonad(0)}$.\\
    \end{enumerate}
		\item If $n=1$ implies  $x,y,z \in \mbox{m}(0)$ and $p_{1}$ is appreciable, then the two first integrals are appreciable.
			\item If $n=0$ implies  $x,y,z \in \mbox{gal(0)}$ and $p_{1}$ is unlimited, then the two first integrals are unlimited.
		\end{enumerate}
		\textbf{Case 3:} If $a$ is unlimited, $a \not\in  G_{1}(G_{2})$, therefore the result of this case is the same as that of Case 2 in proof of Proposition \ref{C}.

\end{proof}

\end{document}